\documentclass[final]{siamltex}
\usepackage[paper=letterpaper,left=1.6in,right=1.6in,top=1in,bottom=1.5in]{geometry}
\usepackage{tikz}
\usepackage{yfonts}
\usetikzlibrary{arrows,snakes,shapes}
\newcounter{xcounter}

\newcounter{cntShader}
\usepackage{color}
\def\couleur{black}
\usepackage{mdwmath}
\usepackage{setspace,color}           
\usepackage[numbers,sort&compress,longnamesfirst,sectionbib]{natbib}
\shortcites{Priezzhev1996,AielloAMR93,GhoshLMMPRRTZ99}
\usepackage{amsmath}
\usepackage{amssymb}
\usepackage{xspace}
\usepackage[latin1]{inputenc}
\usepackage{graphicx}
\usepackage{epsfig}
\usepackage{multirow}
\usepackage{algorithm}
\usepackage{algorithmic}
\newcommand{\dist}{\operatorname{dist}}
\newcommand{\diam}{\operatorname{diam}}

\def\mod{\operatorname{mod}}
\renewcommand{\O}{\mathcal{O}}
\newtheorem{thm}{Theorem}  \newtheorem{lem}[thm]{Lemma}

\newtheorem{cor}[thm]{Corollary}
\newtheorem{fac}[thm]{Fact}

\newtheorem{pro}[thm]{Proposition}

\newtheorem{defi}[thm]{Definition}
\numberwithin{thm}{section}
\newcommand{\thmref}[1]{Theorem~\ref{thm:#1}}
\newcommand{\proref}[1]{Proposition~\ref{pro:#1}}
\newcommand{\lemref}[1]{Lemma~\ref{lem:#1}}
\newcommand{\corref}[1]{Corollary~\ref{cor:#1}}
\newcommand{\defref}[1]{Definition~\ref{def:#1}}
\newcommand{\figref}[1]{Figure~\ref{fig:#1}}
\newcommand{\secref}[1]{Section~\ref{sec:#1}}
\newcommand{\secrefs}[2]{Sections~\ref{sec:#1} and~\ref{sec:#2}}
\newcommand{\facref}[1]{Fact~\ref{fac:#1}}
\newcommand{\eq}[1]{equation~\eqref{eq:#1}}
\newcommand{\eqss}[3]{equations~\eqref{eq:#1},~\eqref{eq:#2}, and~\eqref{eq:#3}}
\newcommand{\mypro}[2]{\begin{pro}\label{pro:#1}#2\end{pro}}
\newcommand{\againpronofrom}[2]{\noindent\textbf{\proref{#1}}\textbf{.}\emph{#2}\par}
\newcommand{\simplePr}{\ensuremath{\operatorname{\mathbf{Pr}}}}
\renewcommand{\Pr}[1]{\ensuremath{\operatorname{\mathbf{Pr}}\left[#1\right]}}
\newcommand{\Pro}[1]{\ensuremath{\operatorname{\mathbf{Pr}}\left[#1\right]}}
\newcommand{\Ex}[1]{\ensuremath{\operatorname{\mathbf{E}}\left[#1\right]}}

\newcommand{\avgx}{\ensuremath{\bar{x}}}
\newcommand{\cG}{\ensuremath{\mathcal{G}}}
\newcommand{\bP}{\ensuremath{\mathbf{P}}}
\newcommand{\bL}{\ensuremath{\mathbf{L}}}
\newcommand{\tP}{\ensuremath{\widetilde{\mathbf{P}}}}
\newcommand{\barP}{\ensuremath{\overline{\mathbf{P}}}}
\newcommand{\ff}{\ensuremath{\mathbf{f}}}
\newcommand{\FF}{\ensuremath{\mathcal{F}}}
\newcommand{\im}{\textswab{i}}
\newcommand{\Oh}{\mathcal{O}}
\newcommand{\ARR}{\mathsf{ARR}}
\newcommand{\Geo}{\mathsf{Geo}}
\newcommand{\tsum}{\textstyle\sum}
\newcommand{\polylog}{\operatorname{polylog}}
\DeclareSymbolFont{AMSb}{U}{msb}{m}{n}
\newcommand{\N}{{\mathbb{N}}}
\newcommand{\Z}{{\mathbb{Z}}}
\newcommand{\Znd}{{\mathbb{Z}^d_{\geq 3}}}
\newcommand{\Nnd}{{\mathbb{N}^d_{\geq 3}}}
\newcommand{\R}{{\mathbb{R}}}

 \newcommand{\FOCS}[2]{#1 Annual IEEE Symposium on Foundations of Computer Science (FOCS'#2)}
 \newcommand{\STOC}[2]{#1 Annual ACM Symposium on Theory of Computing (STOC'#2)}
 \newcommand{\SODA}[2]{#1 Annual ACM-SIAM Symposium on Discrete Algorithms (SODA'#2)}
 \newcommand{\PODC}[2]{#1 Annual ACM Principles of Distributed Computing (PODC'#2)}

\newcommand{\SPAA}[2]{#1 ACM Symposium on Parallel Algorithms and Architectures (SPAA'#2)}

\def\poly{\operatorname{poly}}
\let\oldsqrt\sqrt
\def\hksqrt{\mathpalette\DHLhksqrt}
\def\DHLhksqrt#1#2{\setbox0=\hbox{$#1\oldsqrt{#2\,}$}\dimen0=\ht0
   \advance\dimen0-0.2\ht0
   \setbox2=\hbox{\vrule height\ht0 depth -\dimen0}   {\box0\lower0.4pt\box2}}
\renewcommand\sqrt\hksqrt
\renewcommand{\leq}{\leqslant}
\renewcommand{\geq}{\geqslant}
\renewcommand\epsilon\varepsilon
\newcount\minute \newcount\hour \newcount\hourMins
\def\now{\minute=\time \hour=\time \divide \hour by 60 \hourMins=\hour \multiply\hourMins by 60
  \advance\minute by -\hourMins \zeroPadTwo{\the\hour}:\zeroPadTwo{\the\minute}}

\def\today{\the\year-\zeroPadTwo{\the\month}-\zeroPadTwo{\the\day}}
\def\zeroPadTwo#1{\ifnum #1<10 0\fi #1}
\makeatletter
\DeclareRobustCommand{\qed}{\ifmmode\mathqed\else\leavevmode\unskip\penalty9999\hbox{}\nobreak\hfill\quad\hbox{\qedsymbol}\fi}
\let\QED@stack\@empty
\let\qed@elt\relax
\newcommand{\pushQED}[1]{\toks@{\qed@elt{#1}}\@temptokena\expandafter{\QED@stack}\xdef\QED@stack{\the\toks@\the\@temptokena}}
\newcommand{\popQED}{\begingroup\let\qed@elt\popQED@elt \QED@stack\relax\relax\endgroup}
\def\popQED@elt#1#2\relax{#1\gdef\QED@stack{#2}}
\newif\ifmeasuring@
\newif\iffirstchoice@ \firstchoice@true
\def\setQED@elt#1#2\relax{\ifmeasuring@\else \iffirstchoice@ \gdef\QED@stack{\qed@elt{}#2}\fi\fi#1}
\def\linebox@qed{\hfil\hbox{\qedsymbol}\hfilneg}
\def\math@qedhere{\@ifundefined{\@currenvir @qed}{\qed@warning\quad\hbox{\qedsymbol}}{\@xp\aftergroup\csname\@currenvir @qed\endcsname}}
\def\displaymath@qed{\relax\ifmmode\ifinner\aftergroup\linebox@qed\else\eqno\let\eqno\relax \let\leqno\relax \let\veqno\relax\hbox{\qedsymbol}\fi\else\aftergroup\linebox@qed\fi}
\@xp\let\csname equation*@qed\endcsname\displaymath@qed
\def\equation@qed{
  \iftagsleft@\hbox{\phantom{\quad\qedsymbol}}\gdef\alt@tag{\rlap{\hbox to\displaywidth{\hfil\qedsymbol}}\global\let\alt@tag\@empty}
  \else\gdef\alt@tag{\global\let\alt@tag\@empty\vtop{\ialign{\hfil####\cr\tagform@\theequation\cr\qedsymbol\cr}}\setbox\z@}
  \fi
}
\def\qed@tag{\global\tag@true \nonumber&\omit\setboxz@h {\strut@ \qedsymbol}\tagsleft@false\place@tag@gather\kern-\tabskip\ifst@rred \else \global\@eqnswtrue \fi \global\advance\row@\@ne \cr}
\def\split@qed{\def\endsplit{\crcr\egroup \egroup \ctagsplit@false \rendsplit@\aftergroup\align@qed}}
\def\align@qed{\ifmeasuring@ \tag*{\qedsymbol}\else \let\math@cr@@@\qed@tag\fi}
\@xp\let\csname align*@qed\endcsname\align@qed
\@xp\let\csname gather*@qed\endcsname\align@qed
\def\@tempb#1 v#2.#3\@nil{#2}
\ifnum\@xp\@xp\@xp\@tempb\csname ver@amsmath.sty\endcsname v0.0\@nil<\tw@\def\@tempa{TT}\else\def\@tempa{TF}\fi
\if\@tempa\renewcommand{\math@qedhere}{\quad\hbox{\qedsymbol}}\fi
\newcommand{\openbox}{\leavevmode\hbox to.77778em{\hfil\vrule\vbox to.675em{\hrule width.6em\vfil\hrule}\vrule\hfil}}
\DeclareRobustCommand{\textsquare}{\begingroup\usefont{U}{msa}{m}{n}\thr@@\endgroup}
\providecommand{\qedsymbol}{\openbox}
\renewenvironment{proof}[1][\proofname]{\par\pushQED{\qed}\normalfont\topsep6\p@\@plus6\p@\relax\trivlist\item[\hskip\labelsep\itshape #1\@addpunct{.}]\ignorespaces}{\popQED\endtrivlist\@endpefalse}
\providecommand{\proofname}{Proof}
\makeatother
\allowdisplaybreaks[1]   
\begin{document}

\title{Quasirandom Load Balancing\thanks{A preliminary conference
version~\cite{FGS10}
appeared in the 21st ACM-SIAM Symposium on Discrete Algorithms (SODA 2010).
A full version~\cite{SICOMP1} appeared in SIAM Journal on Computing.
This work was done while all authors
    were postdoctoral fellows at the International Computer Science Institute (ICSI), Berkeley supported
    by the German Academic Exchange Service (DAAD).
}}

\author{
    Tobias Friedrich\thanks{Friedrich-Schiller-Universit?t Jena, Germany.}
\and
    Martin Gairing\thanks{University of Liverpool, United Kingdom.}
\and
    Thomas Sauerwald\thanks{Max-Planck-Institut f\"ur Informatik, Saarbr\"ucken, Germany.}
    }

\date{}

\maketitle

\begin{abstract} \small\baselineskip=9pt
    We propose a simple distributed algorithm
    for balancing indivisible tokens on graphs.
    The algorithm is completely deterministic,
    though it tries to imitate (and enhance) a randomized algorithm by
    keeping the accumulated rounding errors as small as possible.

    Our new algorithm surprisingly closely approximates
    the idealized process (where the tokens are divisible) on important
    network topologies.     On $d$-dimensional torus graphs with $n$ nodes
    it deviates from the idealized
    process only by an additive constant.
    In contrast,
    the randomized rounding approach of Friedrich and Sauerwald~\citep{FS09}      can deviate up to~$\Omega(\polylog(n))$
    and the deterministic algorithm of Rabani, Sinclair and Wanka~\citep{RSW98}     has a deviation of~$\Omega(n^{1/d})$.
    This makes our quasirandom algorithm the first known algorithm
    for this setting which is optimal both in time and achieved smoothness.
    We further show that on
    the hypercube as well, our algorithm has a smaller deviation from
    the idealized process than the previous algorithms.

    To prove these results, we derive several combinatorial and probabilistic
    results that we believe to be of independent interest.
    In particular, we show that first-passage probabilities of a random walk on a
    path with arbitrary weights can be expressed as a convolution of independent geometric
    probability distributions.

\end{abstract}

\section{Introduction}

Load balancing is a requisite for the efficient utilization of
computational resources in parallel and distributed systems. The aim is to
reallocate the load such that afterward, each node has approximately the
same load. Load balancing problems have various applications, e.g., for
scheduling~\citep{Surana06}, routing~\citep{Cyb89}, and numerical computation~\citep{Zhanga09,Williams91}.

Typically, load balancing algorithms iteratively exchange load along edges of an
undirected connected graph.  In the natural \emph{diffusion paradigm}, an arbitrary amount of load
can be sent along each edge at each step~\cite{RSW98,MuthukrishnanGS98}.
For the \emph{idealized} case of divisible
load, a popular diffusion algorithm is the first-order-scheme
by \citet{SubramanianScherson94} whose convergence rate is fairly well captured in terms of the spectral gap~\citep{Lovasz93random}.

However, for many applications the assumption of divisible load may be invalid.
Therefore, we consider the \emph{discrete} case where the load can only be decomposed
into indivisible unit-size tokens. A very natural question is how much this
integrality assumption decreases the efficiency of load balancing. In fact, finding
a precise quantitative relationship
 between the discrete and the idealized case is an open problem posed by many authors, e.g., \citep{GhoshLMMPRRTZ99,GM96,LovaszWinkler95,MuthukrishnanGS98,
SubramanianScherson94,EMS06,FS09,RSW98}.

A simple method for approximating the idealized process was analyzed by \citet*{RSW98}.
Their approach (which we will call ``RSW algorithm'')
is to round down the fractional flow of the idealized process.
They introduce a very useful parameter of the graph called \emph{local divergence}
and prove that it gives tight upper bounds on the deviation between the idealized process
and their discrete process. However,
one drawback of the RSW algorithm is that it can end up in
rather unbalanced states (cf.~\proref{lowerbound}).
To overcome this problem, Friedrich and Sauerwald
 analyzed a new algorithm based on randomized rounding \citep{FS09}. On many graphs, this algorithm approximates the idealized case much better than RSW algorithm's approach of always rounding down.
A natural question is whether this randomized algorithm can be derandomized
without sacrificing on its performance.
For the graphs considered in this work, we answer this question in the affirmative. We introduce a  \emph{quasirandom load balancing algorithm} which rounds
up or down deterministically such that the accumulated rounding errors
on each edge are minimized.
Our approach follows the concept of quasirandomness as it deterministically imitates the
expected behavior of its random counterpart.

\paragraph{Our Results}
We focus on two network topologies: hypercubes and torus graphs.
Both have been intensively studied in the
context of load balancing (see e.g.,~\citep{RSW98,FS09,JH03,P89,GPR99}).
We measure the smoothness
of the load by the so-called \emph{discrepancy} (see e.g.~\cite{RSW98,FS09,GhoshLMMPRRTZ99,EMS06})
which is the difference between the maximum
and minimum load among all nodes.

For \emph{$d$\nobreakdash-dimensional torus graphs}
we prove that
our quasirandom algorithm approximates the idealized process up to
an additive constant (\thmref{torus}).
More precisely, for all initial load distributions and
time steps, the load of any vertex in the discrete process differs from
the respective load in the idealized process only by a constant.
This holds even for non-uniform torus graphs with different side lengths
(cf.~\defref{torus}).
For the uniform torus graph our results are
to be compared with a deviation
 of $\Omega(\polylog(n))$
for the randomized rounding approach (\thmref{toruslower})
and
$\Omega(n^{1/d})$ for the RSW algorithm (\proref{lowerbound}).
Hence, despite the fact that our approach is deterministic, it still improves
over its random counterpart.
Starting with an initial discrepancy of $K$,
the idealized process reaches a constant discrepancy after $\Oh(n^{2/d}\,\log(Kn))$ steps
(cf.~\corref{ideal}).
Hence the same holds for our quasirandom algorithm, which
makes it the first algorithm for the discrete case which is optimal both
in time and discrepancy up to constant factors.

For the \emph{hypercube}, we prove that the deviation
of our quasirandom algorithm
from the idealized process is between $\Omega(\log n)$ and $\Oh(\log^{3/2} n)$
(\thmref{cube}).
Note that the analysis for the upper bound in this paper fixes a bug in the corresponding proof of the conference version~\cite{FGS10},
where we claimed an upper bound of $\Oh(\log n)$.  
For the hypercube
we also show that the deviation of the random approach
is $\Omega(\log n)$
(\thmref{cuberandomlower})
while the deviation of the RSW algorithm is $\Omega(\log^2 n)$
(\proref{lowerbound}).
Again, our quasirandom algorithm is substantially better than
the RSW algorithm~\citep{RSW98}.

\paragraph{Our Techniques}
Instead of analyzing our quasirandom algorithm directly,
we examine a new generic class of load balancing algorithms that
we call \emph{bounded error diffusion} (BED).
Roughly speaking, in a BED algorithm the \emph{accumulated} rounding error on each edge
is bounded by some constant at all times.
This class includes our quasirandom algorithm.

The starting point of \cite{RSW98} and \cite{FS09} as well as that of our paper is to express the deviation
from the idealized case by a certain sum of weighted rounding errors (\eq{StandardAnsatz}).
In this sum, the rounding errors are weighted by transition probabilities of a certain random walk.
Roughly speaking, \citet{RSW98} estimate this sum directly by adding up all transition probabilities.
In the randomized approach of \cite{FS09}, the sum is bounded by Chernoff-type inequalities relying
on independent rounding decisions. We take a completely different approach and prove that the
transition probabilities between two fixed vertices are unimodal in time (cf.~\thmref{cubeunimodal} for the hypercube).
This allows us to upper bound the complete sum by its maximal summand (\lemref{betragkleinerk}) for BED algorithms.
The intriguing combinatorial property of {\em unimodality} is the heart of our proof and seems to be
the main reason why we can outperform the previous approaches.
Even though unimodality has a one-line definition,
it has become apparent that proving it can be
a very challenging task requiring
intricate combinatorial constructions or refined mathematical tools~(see
e.g.\ \citeauthor{Stanley1989}'s survey~\citep{Stanley1989}).

It turns out that this is also the case for the transition probabilities of torus graphs and hypercubes considered here.
The reason is that explicit formulas seem to be intractable, and
typical approximations (e.g.\ Poissonization~\citep{DGM90})
are far too loose to compare consecutive transition probabilities.
For the $d$\nobreakdash-dimensional torus, we use a
local central limit theorem to approximate the transition probabilities by a
multivariate normal distribution which is known to be unimodal.

On hypercubes, the above method fails, as several inequalities for the torus
graph are only true for constant~$d$.
However, we can employ the additional symmetries
of the hypercube
to prove
unimodality of the transition probabilities by relating it to a random walk on a weighted path. Somewhat surprisingly,
this intriguing property was previously unknown, although random walks on hypercubes have been intensively studied~(see e.g.~\cite{DGM90,KLY93,MS04}).

We prove this unimodality result by establishing an interesting result
concerning first-passage probabilities
of a random walk
on paths with arbitrary transition
probabilities: If the loop probabilities are at least~$1/2$, then the
first-passage probability distribution can be expressed as a convolution of independent
geometric distributions. In particular, this implies that these probabilities
are log-concave. Reducing the random walk on a hypercube to a random walk on a
weighted path, we obtain the result that the transition probabilities on the hypercube are unimodal.
Estimating the maximum probabilities via a balls-and-bins-process, we finally obtain our upper bound on the deviation for the hypercube.

We believe that our probabilistic result for paths is of
independent interest, as random walks on the paths are among the most
extensively studied stochastic processes. Moreover, many analyses of randomized
algorithms can be reduced to such random walks (see e.g.~\cite[Thm.~6.1]{MR95}).

\paragraph{Related Work}

In the approach of \citet{ES10}
certain interacting random walks are used to
reduce the load deviation.
This randomized algorithm achieves a constant additive error between the maximum and average load on hypercubes and torus graphs
in time $\Oh( \log (Kn)/(1-\lambda_2))$, where $\lambda_2$ is the second largest eigenvalue of the diffusion matrix. However, in contrast to our deterministic algorithm, this algorithm is less natural and more complicated (e.g., the nodes must have an accurate estimate of the average load).

\citet{AielloAMR93}
and \citet{GhoshLMMPRRTZ99} studied balancing algorithms where, in each time step,
at most one token is transmitted over each edge.
Due to this restriction, these
algorithms take substantially more time, i.e., they run in time at least linear in the initial
discrepancy~$K$.
Nonetheless, the best known bounds on the discrepancy are only
polynomial in~$n$ for the torus and $\Omega(\log^5 n)$ for the hypercube~\cite{GhoshLMMPRRTZ99}.

In another common model, nodes are only allowed to exchange load
with at most one neighbor in each time step, see e.g.,~\cite{GM96,RSW98,FS09}. In fact,
the aforementioned randomized rounding approach \cite{FS09} was analyzed in this model.
However,
the idea of randomly rounding the fractional flow such that the expected error is
zero naturally extends to our diffusive setting where a node may exchange load with all neighbors simultaneously.

\emph{Quasirandomness} describes a deterministic process which imitates certain
properties of a random process.  Our quasirandom load balancing algorithm
imitates the property that rounding up and down the flow between two vertices
occurs roughly equally often, using a deterministic process which minimizes these
rounding errors directly.  In this way, we keep the desired property that the ``expected''
accumulated rounding error is zero, but remove almost all of its (undesired) variance.
Similar concepts have been used for deterministic random walks~\cite{CooperSpencer},
external mergesort~\citep{BarveGV97},
and quasirandom rumor spreading~\cite{DFS08}.  The latter work presents a quasirandom
algorithm which is able to broadcast a piece of information at least as fast as
its random counterpart on the hypercube and most random graphs.  However,
in the case of rumor spreading, the quasirandom protocol is just slightly faster than the random protocol,
while the quasirandom load-balancing algorithm presented here
substantially outperforms its random counterpart.

\paragraph{Organization of the paper}

In \secref{algorithms}, we give a description of our bounded error diffusion
(BED) model. For a better comparison, we present some results for
the previous algorithms of \cite{FS09} and \cite{RSW98} in \secref{others}.
In \secref{basic}, we introduce our basic method
which is used in \secrefs{cube}{torus} to analyze BED algorithms
on hypercubes and torus graphs, respectively.

\section{Model and algorithms}
\label{sec:algorithms}

We aim for balancing load on a connected, undirected graph $G=(V,E)$.
Denote by $\deg(i)$ the \emph{degree}  of node $i\in V$ and let
$\Delta=\Delta(G)=\max_{i\in V} \deg(i)$
be the maximum degree of~$G$.
The balancing process is governed by an ergodic, doubly-stochastic diffusion matrix~$\bP$
with
\[
    \bP_{i,j} =
    \begin{cases}
    \tfrac{1}{2 \Delta} & \text{if $\{i,j\} \in E$,}\\
    1-\tfrac{\deg(i)}{2\Delta} & \text{if $i=j$,}\\
    0 & \text{otherwise.}
    \end{cases}
\]
Let $x^{(t)}$ be the load-vector of the vertices at step~$t$
(or more precisely, after the completion of the balancing procedure at step~$t$).
The \emph{discrepancy}
of such a (row) vector $x$ is $\max_{i,j} ( x_i - x_j )$,
and the discrepancy at step~$0$ is called initial discrepancy~$K$.

\paragraph{The idealized process}
In one time step
each pair $(i, j)$ of adjacent vertices shifts
divisible
tokens between~$i$ and~$j$.
We have the following iteration, $x^{(t)} = x^{(t-1)} \bP$ and inductively,
$x^{(t)} = x^{(0)} \bP^{t}$. Equivalently, for any edge $\{i,j\} \in E$ and step~$t$, the flow from~$i$ to~$j$ at step~$t$ is
$\bP_{i,j} x_i^{(t-1)} - \bP_{j,i} x_j^{(t-1)}$. Note that the symmetry of~$\bP$ implies that for $t \rightarrow \infty$, $x^{(t)}$ converges towards the uniform vector $(1/n,1/n,\ldots,1/n)$.

\paragraph{The discrete process}
There are different ways to handle non-divisible tokens.
We define the following \emph{bounded error diffusion} (BED) model.
Let $\Phi_{i,j}^{(t)}$ denote the integral flow from~$i$ to~$j$ at time~$t$.
As $\Phi_{i,j}^{(t)}=-\Phi_{j,i}^{(t)}$, we have
$x_i^{(t)} = x_i^{(t-1)} - \sum_{j\colon\{i,j\}\in E} \Phi_{i,j}^{(t)}$.
Let
$e_{i,j}^{(t)} := \big(\bP_{i,j} x_i^{(t-1)} - \bP_{j,i} x_j^{(t-1)}\big) - \Phi_{i,j}^{(t)}$
be the excess load allocated to~$i$ as a result
of rounding on edge $\{i,j\}$ in time step~$t$.
A negative value of $e_{i,j}^{(t)}$ signifies a deficit of load.
Note that for all vertices~$i$,
$x_i^{(t)} = (x^{(t-1)} \bP)_i + \sum_{j\colon\{i,j\}\in E} e_{i,j}^{(t)}$.
Now, let~$\Lambda$ be an upper bound for
the accumulated rounding errors (deviation from the idealized process),
that is,
$\big|\sum_{s=1}^t e_{i,j}^{(s)} \big|\leq \Lambda$ for all $t \in \N$
and $\{i,j\}\in E$.
All our bounds still hold if~$\Lambda$ is a function of~$n$ and/or~$t$,
but we only say that an algorithm is a \emph{BED algorithm} if~$\Lambda$
is a constant.

For $\bP_{i,j} x_i^{(t)} \geq \bP_{j,i} x_j^{(t)}$,
our new \emph{quasirandom diffusion algorithm} chooses 
the flow $\Phi_{i,j}^{(t)}$ from~$i$ to~$j$ to be either
$\Phi_{i,j}^{(t)}=\big\lfloor \bP_{i,j} x_i^{(t)} - \bP_{j,i} x_j^{(t)} \big\rfloor$
or
$\Phi_{i,j}^{(t)}=\big\lceil \bP_{i,j} x_i^{(t)} - \bP_{j,i} x_j^{(t)} \big\rceil$
such that
$\big|\sum_{s=1}^t e_{i,j}^{(s)} \big|$
is minimized.
This yields a BED algorithm with $\Lambda\leq1/2$, which
can be implemented
with $\Oh(\log \Delta)$ storage per edge.
Note that one can imagine various other natural (deterministic or randomized)
BED algorithms.  To achieve this, the algorithm only
has to ensure that the errors do not add up to more than a constant.

With the above notation,
the \emph{RSW algorithm} uses
$\Phi_{i,j}^{(t)}=\big\lfloor \bP_{i,j} x_i^{(t)} - \bP_{j,i} x_j^{(t)} \big\rfloor$,
provided that $\bP_{i,j} x_i^{(t)} \geq \bP_{j,i} x_j^{(t)}$.
In other words, the flow on each edge is always rounded down.
In our BED framework this would imply a~$\Lambda$ of order~$T$ after~$T$ time steps.

A simple \emph{randomized rounding diffusion algorithm}
chooses for $\bP_{i,j} x_i^{(t)} \geq \bP_{j,i} x_j^{(t)}$ the flow
$\Phi_{i,j}^{(t)}$ as the randomized rounding of
$\bP_{i,j} x_i^{(t)} - \bP_{j,i} x_j^{(t)}$, that is, it
rounds up with probability $(\bP_{i,j} x_i^{(t)} - \bP_{j,i} x_j^{(t)}) -
\big\lfloor \bP_{i,j} x_i^{(t)} - \bP_{j,i} x_j^{(t)} \big\rfloor$ and
rounds down otherwise.
This typically achieves an error
$\Lambda$ of order $\sqrt{T}$ after~$T$ time steps.

\paragraph{Handling Negative Loads}

Unless there is a lower bound on the minimum load of a vertex,
negative loads may occur during the balancing procedure.
In what follows, we describe a simple approach for coping with this problem.

Consider a graph~$G$ for which we can prove a deviation of at most
$\gamma$ from the idealized process. Let $x^{(0)}$ be the initial load vector
with an average load of~$\avgx$. Then at the beginning of the balancing
procedure, each node generates~$\gamma$ additional (virtual) tokens. During the
balancing procedure, these tokens are regarded as common tokens, but at the end
they are ignored. First observe that since the minimum load at each node in the idealized
process is at least~$\gamma$, it follows that at each step, every node has at
least a load of zero in the discrete process. Since each node has a load of
$\avgx + \Oh(\gamma)$ at the end, we end up with a load distribution where the
maximum load is still $\avgx + \Oh(\gamma)$ (ignoring the virtual tokens).

\section{Basic method to analyze our quasirandom algorithm}
\label{sec:basic}

To bound runtime and discrepancy of a BED algorithm,
we always bound the deviation between the idealized model and the discrete model,
which is an important measure in its own right.
For this discussion, let $x_\ell^{(t)}$ denote the load on vertex~$\ell$ in step~$t$ in the discrete model and
$\xi_\ell^{(t)}$ denote the load on vertex~$\ell$ in step~$t$ in the idealized model. We assume that the discrete and idealized model start with the same initial load, that is, $x^{(0)}=\xi^{(0)}$.
As derived in \citet{RSW98}, their difference can be written as
\begin{align}
    \label{eq:StandardAnsatz}
    x_\ell^{(t)} - \xi_\ell^{(t)}
    =
    \sum_{s=0}^{t-1}
    \sum_{[i:j]\in E}
    e_{i,j}^{(t-s)}
    (\bP_{\ell,i}^s - \bP_{\ell,j}^s).
\end{align}
where $[i:j]$ refers to an edge $\{i,j\}\in E$ with $i < j$,  and ``$<$'' is some arbitrary but fixed ordering on the vertices~$V$.
It will be sufficient to bound \eq{StandardAnsatz}, as the convergence speed of the idealized process can be bounded in terms of the second largest eigenvalue.
\begin{thm}[{e.g., \cite[Thm.~1]{RSW98}}]\label{thm:ideal}
    On all graphs with second largest eigenvalue in absolute value $\lambda_2=\lambda_2(\bP)$,
    the idealized process
    with divisible tokens
        reduces
    an initial
    discrepancy~$K$
    to~$\ell$ within
    \[
       \tfrac{2}{1-\lambda_2} \ln \big( \tfrac{K n^2}{\ell} \big)
    \]
    time steps.
\end{thm}

As $\lambda_2=1-\Theta(\log^{-1} n)$ for the hypercube and
$\lambda_2=1-\Theta(n^{-2/d})$ for the
$d$\nobreakdash-dimensional torus~\cite{GM96}, one immediately obtains the following corollary.
\begin{cor}\label{cor:ideal}
    The idealized process
    reduces an initial
    discrepancy of~$K$  to~$1$ within
    \(
       \O(n^{2/d}  \, \log(K n) )
    \)
    time steps on the $d$\nobreakdash-dimensional torus
    and within
    \(
        \O(\log n \, \log(K n) )
    \)
    time steps
    on the hypercube.
\end{cor}

An important property of all examined graph classes will be
unimodality or log-concavity of certain transition probabilities.
\begin{defi}
    A function $f\colon\N\to \R_{\ge 0}$ is  \emph{log-concave}
    if $f(i)^2\ge f(i-1) \cdot f(i+1)$ for all $i\in\N_{>0}$.
\end{defi}

\begin{defi}
    A function $f\colon\N \to \R$ is \emph{unimodal}
    if there is a $s \in \N$ such that
    $f|_{x \le s}$ as well as $f|_{x \ge s}$ are monotone.
\end{defi}

Log-concave functions are sometimes also called \emph{strongly unimodal}~\citep{KeilsonGerber1971}.
We summarize some classical results regarding log-concavity and unimodality.
\begin{fac}
    \label{fac:unimodal}
    \begin{enumerate}
    \item Let~$f$ be a log-concave function.
          Then, $f$ is also a unimodal function (e.g.~\citep{Keilson,KeilsonGerber1971}).
    \item Hoggar's theorem~\citep{Hoggar1974}:
          Let $f$ and~$g$ be log-concave functions.
          Then their convolution           $(f*g)(k)=\sum_{i=0}^k f(i) \, g(k-i)$
          is also log-concave.
    \item Let $f$ be a log-concave function and $g$ be a unimodal function.
          Then their convolution $f*g$ is a unimodal function~\citep{KeilsonGerber1971}.
    \end{enumerate}
\end{fac}

Our interest in unimodality is based on the fact that an alternating sum over a
unimodal function can be bounded by its maximum.  More precisely, for
a non-negative and unimodal function $f\colon X \to \R$
and $t_0, \ldots, t_k \in X$ with $t_0 \leq \cdots \leq t_k$, the following holds:
\[
    \bigg| \sum_{i = 0}^k (-1)^i \, f(t_i)\bigg| \le \max_{x \in X} f(x).
\]
This is a special case of Abel's inequality.  We generalize both in the following lemma.

\begin{lem}
    \label{lem:betragkleinerk}
    Let $f\colon X \to \R$ be non-negative with $X\subseteq\R$.
    Let $A_0, \ldots, A_k \in \R$ and
    $t_0, \ldots, t_k \in X$ such that $t_0 \leq \cdots \leq t_k$
    and
    $|\sum_{i=a}^{k} A_i| \leq k$ for all $0\leq a\leq k$.
    If $f$ has $\ell$ local extrema, then
    \[
        \bigg| \sum_{i = 0}^k A_i \, f(t_i)\bigg| \
        \le\  (\ell+1)\,k\,\max_{j=0}^k f(t_j).
    \]
\end{lem}
\begin{proof}
    Let us start with the assumption that $f(t_i)$, $0\leq i\leq k$, is monotone increasing.
    With $f(t_{-1}):=0$, it then is easy to see that
    \begin{align*}
        \big|\tsum_{i = 0}^k A_i \, f(t_i)\big|
        &= \big|\tsum_{i = 0}^k \sum_{j=0}^i A_i \, (f(t_j)-f(t_{j-1}))\big|\\
        &= \big|\tsum_{j = 0}^k \sum_{i=j}^k A_i \, (f(t_j)-f(t_{j-1}))\big|\\
        &\leq \tsum_{j = 0}^k \big| f(t_j)-f(t_{j-1})\big| \, \big| \sum_{i=j}^k A_i \big|\\
        &\leq \tsum_{j = 0}^k \big| f(t_j)-f(t_{j-1}) \big|\, k\\
        &= k\,\max_{j=0}^k f(t_j).
        \end{align*}
    The same holds if $f(t_i)$, $0\leq i\leq k$, is monotone decreasing.
    If $f(x)$ has $\ell$ local extrema, we split the sum in $(\ell+1)$ parts
    such that $f(x)$ is monotone on each part and apply the above arguments.
\end{proof}

\paragraph{Random Walks}
To examine the diffusion process, it will be useful to define a random walk based on~$\bP$.
For any pair of vertices $i,j$, $\bP^{t}_{i,j}$ is
the probability that a random walk guided by~$\bP$ starting from~$i$
is located at~$j$ at step~$t$.
In \secref{cube} it will be useful to set $\bP_{i,j}(t):=\bP^{t}_{i,j}$
and to denote with $\ff_{i,j}(t)$ for $i \neq j$
the first-passage probabilities, that is,
the probability that a random walk starting from $i$ visits
the vertex $j$ at step~$t$ for the first time.

\section{Analysis on the hypercube}
\label{sec:cube}

We first give the definition of the hypercube.
\begin{defi}
A $d$\nobreakdash-dimensional hypercube with $n=2^d$ vertices
has vertex set $V=\{0,1\}^d$ and edge set
$E=\{ \{i,j\} \mid \text{$i$ and~$j$ differ in one bit}\}$.
\end{defi}

In this section we prove the following result.

\begin{thm}
    \label{thm:cube}
    For all initial load vectors
    on the $d$\nobreakdash-dimensional hypercube with $n$~vertices,
    the deviation between the idealized process and
    a discrete process with accumulated rounding errors at most~$\Lambda$
    is upper bounded by $\Oh( \Lambda \cdot (\log n)^{3/2})$ at all times, and there
    are load vectors for which this deviation can be $(\log n)/2$.
\end{thm}

Recall that for BED algorithms $\Lambda=\Oh(1)$.
With \thmref{ideal} it follows that any BED algorithm (and in particular
our quasirandom algorithm) reduces the discrepancy
of any initial load vector with discrepancy~$K$ to $\Oh(\log n)$ within
$\Oh( \log n \, \log (K n) )$ time steps.

\newcommand{\rp}{\beta}
\newcommand{\lp}{\alpha}
\newcommand{\path}{\mathcal{P}}
\newcommand{\pathN}{\path\setminus\{d\}}

\subsection{Log-concave passage time on paths}
To prove \thmref{cube}, we first consider a discrete-time
random walk on a path $\path=(0,1,\ldots, d)$ starting
at node~$0$.
We make use of
a special generating function, called \emph{$z$\nobreakdash-transform}.
The $z$\nobreakdash-transform of a function $g\colon\N\mapsto \R_{\ge 0}$ is defined by
$\cG(z)=\sum_{i=0}^{\infty} g(i)\cdot z^{-i}$.
We will use the fact that a convolution reduces
 to multiplication in the $z$\nobreakdash-plane. More formally, if $\cG_1$ and $\cG_2$ are the 
 $z$\nobreakdash-transforms of $g_1$ and $g_2$, respectively, then the product $\cG_1 \cdot \cG_2$ is the 
 $z$\nobreakdash-transform of  their convolution $g_1 * g_2$.
Instead of the $z$-transform one could carry out a similar analysis using the
\emph{probability generating function}. We choose to use the $z$-transform here since it leads to
slightly simpler arithmetic expressions.

Our analysis also uses
the \emph{geometric distribution} with parameter~$p$, which is defined by $\Geo(p)(t)=(1-p)^{t-1}\,p$ for $t \in \N \setminus \{ 0 \}$
and  $\Geo(p)(0)=0$. It is easy to check that $\Geo(p)$ is log-concave.
Moreover, the $z$\nobreakdash-transform of $\Geo(p)$ is
\[
          \sum_{i=1}^{\infty} \Geo(p)(i) \cdot z^{-i} =  \dfrac{p}{z-(1-p)}.
\]

\noindent
For each node $i\in\path$, let $\lp_i$ be the loop probability at node~$i$ and
$\rp_i$ be the \emph{upward probability}, i.e., the probability of moving to node $i+1$.
Then, the \emph{downward probability} at node~$i$ is $1-\lp_i-\rp_i$.
We can assume that $\rp_i>0$ for all $i\in\pathN$.
We are interested in the first-passage probabilities $\ff_{0,d}(t)$. Observe that
\begin{align}
\label{eq:faltung}
   \ff_{0,d}(t)= (\ff_{0,1} * \ff_{1,2} * \cdots * \ff_{d-1,d}) (t).
\end{align}
In the following argument, we will show that $\ff_{0,d}(t)$ is \emph{log-concave}. Indeed, we show a much stronger result:
\begin{thm}
\label{thm:logconcavePath}
Consider a random walk on a path $\path=(0,1, \ldots, d)$ starting at node~$0$. If $\lp_i\ge\frac{1}{2}$ for
all nodes $i\in\path$, then $\ff_{0,d}$ can be expressed as convolution of $d$~independent
geometric distributions.
\end{thm}

As the geometric distribution is log-concave
and the convolution of log-concave functions is again log-concave (cf.~\facref{unimodal}),
we immediately  get the following corollary.
\begin{cor}
\label{cor:logconcavePath}
Consider a random walk on a path $\path=(0,1, \ldots, d)$ starting at node~$0$. If $\lp_i\ge\frac{1}{2}$ for
all nodes $i\in\path$, then $\ff_{0,d}(t)$ is log-concave in~$t$.
\end{cor}

\medskip
Note that \thmref{logconcavePath} follows directly from
Theorem 1.2 of \citet{Fill2009a}.
As \thmref{logconcavePath} is a crucial ingredient for proving our main result
(\thmref{cube}) for the hypercube, we give a different alternative proof of the statement.
While \citeauthor{Fill2009a}'s proof
is purely stochastic, our proof is based on functional analysis of the $z$\nobreakdash-transform.
Our analysis for the discrete-time random walk should also
be compared with \citeauthor{Keilson}'s analysis of the
continuous-time process~\cite{Keilson}.
The continuous-time process was independently considered even
earlier by \citet{KarlinMcGregor59}.

Before proving the theorem, we will show how to obtain $\ff_{0,d}(t)$ by a recursive argument.
To do this, suppose we are at node~$i\in\pathN$. The next step is a loop with probability~$\lp_i$.
Moreover, the next subsequent non-loop move ends at $i+1$ with probability
$\frac{\rp_i}{1-\lp_i}$ and at $i-1$ with probability
$\frac{1-\rp_i-\lp_i}{1-\lp_i}$.
Thus, for all $i\in \pathN$,
\begin{align*}
  \ff_{i,i+1}(t) = \frac{\rp_i}{1-\lp_i} \cdot \Geo(1-\lp_i)(t)
                           + \frac{1-\rp_i-\lp_i}{1-\lp_i}\cdot (\Geo(1-\lp_i)*\ff_{i-1,i}*\ff_{i,i+1})(t),
\end{align*}
with corresponding $z$\nobreakdash-transform
\begin{align*}
  \FF_{i,i+1}(z) = \frac{\rp_i}{1-\lp_i} \cdot \frac{1-\lp_i}{z-\lp_i}
                          + \frac{1-\rp_i-\lp_i}{1-\lp_i} \cdot \frac{1-\lp_i}{z-\lp_i}
                              \cdot \FF_{i-1,i}(z) \cdot \FF_{i,i+1}(z).
\end{align*}
Rearranging terms yields
\begin{align}
\label{eq:f:recF}
  \FF_{i,i+1}(z) =\frac{ \rp_i}
                                 {                                    z-\lp_i-(1-\rp_i-\lp_i)\cdot \FF_{i-1,i}(z)
                                                                      },
\end{align}
for all $i\in\pathN$. So $ \FF_{i,i+1}(z)$ is obtained recursively with
$   \FF_{0,1}(z) = \frac{\rp_0}{z-(1-\rp_0)}.
$ Finally, the $z$\nobreakdash-transform of \eq{faltung} is
\begin{align}
\label{eq:f:FF0N}
   \FF_{0,d}(z) =  \FF_{0,1}(z) \cdot \FF_{1,2}(z) \cdot \ldots \cdot \FF_{d-1,d}(z).
\end{align}
We will now prove some properties of $ \FF_{i,i+1}(z)$ for $i\in\pathN$.
\medskip
\begin{lem}
    \label{lem:dec}
    Except for singularities, $\FF_{i,i+1}(z)$     is monotone decreasing in $z$,
    for all $i\in\pathN$.
\end{lem}
\begin{proof}
We will show the claim by induction on~$i$.
It is easy to see that the claim holds for the base case ($i=0$) since $\FF_{0,1}(z) = \frac{\rp_0}{z-(1-\rp_0)}$.  Assume inductively that the claim holds for $\FF_{i-1,i}(z)$.
With $1-\rp_i-\lp_i\ge 0$ this directly implies that the denominator
of  \eq{f:recF}
is increasing in~$z$. The claim for $\FF_{i,i+1}(z)$ follows.
\end{proof}
\begin{lem}
    \label{lem:numpoles}
    For all $i\in\pathN$,
    $\FF_{i,i+1}(z)$ has exactly~$i+1$ poles which are all in the interval~$(0,1)$.
    The poles of $\FF_{i,i+1}(z)$
    are distinct from the poles of $\FF_{i-1,i}(z)$.
\end{lem}
\begin{proof}
    Before proving the claims of the lemma, we will show that
    $\FF_{i,i+1}(0)\ge-1$ and $\FF_{i,i+1}(1)=1$ for all $i\in\pathN$. Observe, that
    $\FF_{0,1}(0)=\frac{\rp_0}{-(1-\rp_0)}=\frac{1-\lp_0}{-\lp_0}\ge -1$, since $\lp_0\ge \frac{1}{2}$.
    Also observe that $\FF_{0,1}(1)=1$.
    Assume, inductively that $\FF_{i-1,i}(0)\ge-1$ and $\FF_{i-1,i}(1)=1$. Then with \eq{f:recF},
    $\FF_{i,i+1}(0) \ge \frac{\rp_i}{-\lp_i-(1-\rp_i-\lp_i)\cdot(-1)}
    =\frac{\rp_i}{1-2\lp_i-\rp_i}\ge -1$, since $1-2\lp_i\le 0$.
    Moreover, $\FF_{i,i+1}(1)=\frac{\rp_i}{1-\lp_i-(1-\lp_i-\rp_i)}=1$.
    Thus, $\FF_{i,i+1}(0)\ge-1$ and $\FF_{i,i+1}(1)=1$ for all $i\in\pathN$.

    We will now show the claims of the lemma by induction.
    For the base case, observe that $\FF_{0,1}(z) = \frac{\rp_0}{z-(1-\rp_0)}$ has one pole at
    $z=1-\rp_0>0$ and $\FF_{-1,0}$ is not defined. This implies the claim for $i=0$. Suppose the claim
    holds for $\FF_{i-1,i}(z)$, and let $z_1,z_2, \ldots z_i$ be the poles
    of  $\FF_{i-1,i}(z)$. Without loss of generality, we assume $0<z_1<z_2<\cdots<z_i<1$.
    Let $g_i(z)$ be the denominator of \eq{f:recF}, that is,
    \begin{align*}
        g_i(z):=z-\lp_i-(1-\rp_i-\lp_i)\cdot \FF_{i-1,i}(z).
    \end{align*}
    Observe that
    \begin{itemize}
    \item[(i)]
    $g_i(z)$ has the same set of poles as $\FF_{i-1,i}(z)$,
    \item[(ii)]
    $\lim_{z\rightarrow -\infty} g_i(z)= - \infty$, and
    \item[(iii)]
    $\lim_{z\rightarrow \infty} g_i(z)= \infty$.
    \end{itemize}
    By \eq{f:recF}, $\FF_{i,i+1}(z)$ has its poles at the zeros of
    $g_i(z)$.
   \lemref{dec} shows that in each interval
    $(z_j,z_{j+1})$ with $1\le j \le i-1$,
    $g_i(z)$ is increasing in~$z$. Using
        fact (i)
    this implies that $g_i(z)$ has exactly one zero in each interval $(z_j,z_{j+1})$.
    Thus $\FF_{i,i+1}(z)$ has exactly
    one pole in each interval $(z_j,z_{j+1})$.
    Similarly, \lemref{dec}, together with
        facts (i), (ii) and (iii),
    implies that $\FF_{i,i+1}(z)$ has exactly
    one pole, say $z'$, in the interval $[-\infty,z_1)$
    and
    one pole, say~$z''$, in the interval $(z_i,\infty]$ .
    This implies that $\FF_{i,i+1}(z)$ has exactly $i+1$ poles
    which are all distinct from the poles of $\FF_{i-1,i}(z)$. It remains to
    show that $z'>0$ and $z''<1$.

    Since $\FF_{i-1,i}(0)\ge-1$ and $\lim_{z\rightarrow -\infty}\FF_{i-1,i}(z)=-0$, it
    follows from \lemref{dec} that $-1\le \FF_{i-1,i}(z) \le 0$ for all real $z<0$. So
        $g_i(z)<0$
    for all real $z<0$. It follows that $z'>0$.
    Similarly, since $\FF_{i-1,i}(1)=1$ and  $\lim_{z\rightarrow \infty}\FF_{i-1,i}(z)=+0$,
    it  follows by \lemref{dec} that $0\le \FF_{i-1,i}(z) \le 1$ for all real $z> 1$.
    So
    $g_i(z)>0$ for all real $z>1$. It follows that $z''<1$.
    This finishes the proof of our inductive step. The claim follows.
\end{proof}

The distinctness of the $i+1$ poles of
   $F_{i-1,i}(z)$, established in \lemref{numpoles}, is crucial for the proof of the following lemma.
\begin{lem}
    \label{lem:factor}
    Let $(b_{j,i})_{j=0}^{i}$ be the poles of $\FF_{i,i+1}(z)$, $i\in\pathN$, and
    define $P_i(z)=\prod_{j=0}^{i} (z-b_{j,i})$.
    Then $\FF_{i,i+1}(z)=\rp_i \cdot \frac{P_{i-1}(z)}{P_{i}(z)}$ for all $i\in\pathN$.
\end{lem}
\begin{proof}
    Our proof is by induction on $i$.
    For the base case ($i=0$), observe that $P_{-1}(z)=1$ and thus $\FF_{0,1}(z)$ has the desired form.
    Suppose the claim holds for $\FF_{i-1,i}(z)$. Then \eq{f:recF} implies
    \begin{align}
    \FF_{i,i+1}(z)
       &= \frac{\rp_i}
                     {z-\lp_i-(1-\rp_i-\lp_i)\cdot \rp_{i-1}\cdot\frac{P_{i-2}(z)}{P_{i-1}(z)}}
        \nonumber\\&
        = \frac{\rp_i\cdot P_{i-1}(z)}
            {(z-\lp_i)\cdot P_{i-1}(z)- (1-\rp_i-\lp_i)\cdot
            \rp_{i-1}\cdot P_{i-2}(z)}\label{eq:f:factor1}.
    \end{align}
    Observe that $(z-\lp_i)\cdot P_{i-1}(z)$  is a polynomial of degree~$i+1$ where
    the leading term has a coefficient of~$1$. This also holds for the denominator
    of \eq{f:factor1}, since there, we only subtract a polynomial of order $i-1$.
    By \lemref{numpoles} we know that $\FF_{i,i+1}(z)$ has exactly~$i+1$ real
    positive poles which are all distinct.  
    So the denominator of  \eq{f:factor1} has exactly $d+1$ zeros -- the $b_{j,i}$'s.
    The only polynomial of order $i+1$ that has exactly those zeros and leading coefficient $1$ is 
     $P_{i}(z)$.
    It follows that the denominator of \eq{f:factor1} is equal to
    $P_{i}(z)$.
    The claim follows.
\end{proof}
We are now ready to prove \thmref{logconcavePath}.
\begin{proof}[Proof of \thmref{logconcavePath}]
By \eq{f:FF0N} and
\lemref{factor}, we get
\begin{align*}
\FF_{0,d}(z)
= \prod_{i=0}^{d-1} \FF_{i,i+1}(z)
= \prod_{i=0}^{d-1} \left(\rp_i \cdot \frac{P_{i-1}(z)}{P_{i}(z)}\right)
= \frac{\prod_{i=0}^{d-1} \rp_i}{P_{d-1}(z)}
= K_d
              \cdot \prod_{i=0}^{d-1} \frac{1-b_{i,d-1}}{z-b_{i,d-1}},
\end{align*}
where $(b_{i,d-1})_{i=0}^{d-1}$ are the poles of
$\FF_{d-1,d}(z)$ as defined in \lemref{factor} and
$K_d=\prod_{i=0}^{d-1}\frac{\rp_i}
                 {1-b_{i,d-1}}$.
By \lemref{numpoles}, $b_{i,d-1}\in (0,1)$ for all~$i$.
Now for each~$i$ the term $\frac{1-b_{i,d-1}}{z-b_{i,d-1}}$ is the $z$\nobreakdash-transform of
the geometric distribution with parameter $1-b_{i,d-1}$, i.e., $\Geo(1-b_{i,d-1})(t)$.

Thus, $\ff_{0,d}(t)$ can be expressed as the convolution of~$d$ independent geometric distributions
\begin{align*}
\ff_{0,d}(t) = K_d \cdot [\Geo(1-b_{0,d-1}) * \Geo(1-b_{1,d-1}) * \ldots * \Geo(1-b_{d-1,d-1})](t).
\end{align*}
Moreover, since $\ff_{0,d}$ is a probability distribution over~$t$ and the convolution of probability
distributions is again a probability distribution, we have $K_d=1$.
The theorem follows.
\end{proof}
It should be noted that it follows from \cite[Theorem 1.2]{Fill2009a} that the
parameters $(b_{i,d-1})_{i=0}^{d-1}$ in the geometric distributions are the
eigenvalues of the underlying transition matrix.

Recall that our aim is to prove unimodality for the function $\bP_{0,j}^{t}$ (in $t$). Using the simple convolution formula $\bP_{0,j} = \ff_{0,j} * \bP_{j,j}$ and the log-concavity of $\ff_{0,j}$, it suffices to prove that $\bP_{j,j}$ is unimodal (cf.~\facref{unimodal}). Next, we will prove that $\bP_{j,j}$ is also non-increasing in $t$.

\begin{lem}
\label{lem:monotone}
Let $\bP$ be the $(d+1)\times(d+1)$-transition matrix defining an ergodic Markov chain on a path $\mathcal{P}=(0, \ldots, d)$.
If $\bP_{ii}\ge\frac{1}{2}$ for all $0 \leq i \leq d$
then for all
    $0 \leq i \leq d$, $ \bP_{i,i}^{t} $ is non-increasing in $t$.
\end{lem}
\begin{proof}
It is well known that ergodic Markov chains on paths are time reversible (see e.g. Section 4.8 of \citet{Ro07}).
To see this, let $\pi=(\pi_0, \ldots, \pi_d)$ be the stationary distribution. Then for all
$0\le i\le d-1$ the rate at which the process goes from $i$ to $i+1$ (namely, $\pi_i \bP_{i,i+1}$) is
equal to the rate at which the process goes from $i+1$ to $i$ (namely, $\pi_{i+1} \bP_{i+1,i}$).
Thus, $\bP$ is time-reversible.

One useful property of a time-reversible matrix is that its eigenvalues are all real.
The Ger\u{s}gorin disc theorem states that every eigenvalue $\lambda_j$, $0\le j \le d$,
satisfies the condition
\[
   |\lambda_j - \bP_{ii}| \le 1-\bP_{ii},
\]
for some $0 \le i \le d$.
Since $\bP_{ii}\ge\frac{1}{2}$, this directly implies that
all eigenvalues are in the interval $[0,1]$.

It is well-known
that there is an orthonormal base of $\R^{d+1}$ which is formed by the eigenvectors $v_{0},\,v_{1},\,\ldots,\,v_{d}$ (see e.g.~\cite{Gur00}).
Then for any $n$-dimensional vector $w \in \R^{d+1}$,
$w= \sum_{j=0}^{d} \langle w, v_j \rangle \, v_j$,
where $\langle \thinspace\cdot\thinspace,\thinspace\cdot\thinspace \rangle$
denotes the inner product.
Applying this to the $i$\nobreakdash-th unit vector $e_{i}$ and using $[\thinspace\cdot\thinspace]_{i}$
to denote the $i$\nobreakdash-th entry of a vector in $\R^{d+1}$, we obtain
\[
    e_{i}
    = \sum_{j=0}^{d} \langle e_i, v_j \rangle \, v_j
   = \sum_{j=0}^d [v_{j}]_i v_{j}.
\]
Thus,
\begin{align*}
    \bP^{t} e_{i}
    = \bP^{t} \, \bigg( \sum_{j=0}^d [v_{j}]_i \, v_{j} \bigg)
    = \sum_{j=0}^d [v_{j}]_i \, \bP^{t} \, v_{j}
    = \sum_{j=0}^d [v_{j}]_i \, \lambda^{t}_j v_j
\end{align*}
and finally
\begin{align*}
    \bP_{i,i}^{t}
    = \left[ \bP^{t} e_{i} \right]_{i}
    = \sum_{j=0}^{d} [v_{j}]_i \, \lambda_j^{t} [v_j]_{i}
    = \sum_{j=0}^d \lambda_j^t \, [v_j]_{i}^2,
\end{align*}
which is non-increasing in $t$ as $[v_j]_{i} \in \mathbb{R}$ and $0 \le \lambda_j \le 1$ for all $0\le j \le d$.
\end{proof}

\subsection{Unimodal transition probabilities on the hypercube}
Combining \lemref{monotone} and \thmref{logconcavePath} and then projecting the random walk
on the hypercube on a random walk on a path, we obtain the following result.
\begin{thm}\label{thm:cubeunimodal}
Let $i,j \in V$ be two vertices of a $d$\nobreakdash-dimensional hypercube. Then,
$\bP_{i,j}(t)$ is unimodal.
\end{thm}
\begin{proof}
	We use the following projection of a random walk on a $d$-dimensional hypercube with loop probability $1/2$ to a
    random walk on a path with $d$ vertices, again with loop probability $1/2$.
    The induced random walk is
    obtained from the mapping $x \mapsto |x|_1$, that is, vertices in $\{0,1\}^d$ with
    the same number of ones are equivalent. It is easy to check that this new random
    walk is a random walk on a path with vertices $0,1,\ldots,d$ that moves right
    with probability $\lambda_k=\frac{d-k}{2k}$, left with probability
    $\mu_k=\frac{d}{2k}$, and loops with probability~$\frac{1}{2}$. (This process is also known as the Ehrenfest chain~\cite{GS01}).

    Consider now the random walk on the path with vertex set $\{ 0,1,\ldots,d \}$ and
    let $j$~be an arbitrary number with $0 \leq j \leq d$.
    Recall that $\bP_{0,j}$ can be expressed as a convolution (cf.~\cite{GS01}) of~$\bP$ and~$\ff$ as follows,
    \begin{align*}
       \bP_{0,j} = \ff_{0,j} * \bP_{j,j}.     \end{align*}
    By \corref{logconcavePath},
    $\ff_{0,j}(t)$ is log-concave. Moreover, \lemref{monotone} implies that $\bP_{j,j}(t)$ is non-increasing
    in $t$ and hence unimodal.
    As the convolution
    of any log-concave function with any unimodal function is again unimodal~(cf.~\facref{unimodal}),
    it follows that $\bP_{0,j}(t)$ is unimodal in~$t$.

    Now fix two vertices $i,j$ of the $d$\nobreakdash-dimensional hypercube. By symmetry, we may
    assume that $i=0^d \equiv 0$. Conditioned on the event that the projected random
    walk is located at a vertex with $|j|_1$ ones at step~$t$, every
    vertex with $|j|_1$ ones is equally likely. This gives
    $\bP_{0,j}(t) = \bP_{0,|j|_1}(t) / \binom{d}{|j|_1}$, and therefore the unimodality of
    $\bP_{0,|j|_1}(t)$ implies directly the unimodality of
    $\bP_{0,j}(t)$, as required.
\end{proof}

With more direct methods, one can prove the following supplementary result
giving further insights into the distribution of $\bP_{i,j}(t)$.
As the result is not required for our analysis,
the proof is given in the appendix.
\newcommand{\textpromonotonicity}{
Let $i,j\in V$ be two vertices of the  $d$\nobreakdash-dimensional hypercube with $\dist(i,j) \geq d/2$. Then
$\bP_{i,j}(t)$ is monotone increasing.
}
\mypro{monotonicity}{\textpromonotonicity}

\subsection{Analysis of the discrete algorithm}

We are now ready to prove our main result for hypercubes.
\begin{proof}[Proof of \thmref{cube}]
    By symmetry, it suffices to bound the deviation at the vertex $0 \equiv 0^d$.
    Hence by \eq{StandardAnsatz} we have to bound
    \begin{align*}
        \big| x_0^{(t)} - \xi_0^{(t)} \big|
        &\leq \big|
        \tsum_{s=0}^{t-1}
        \tsum_{[i:j] \in E} e_{i,j}^{(t-s)} (\bP_{0,i}(s) - \bP_{0,j}(s)) \big| \\
        &\leq \big|
        \tsum_{s=0}^{t-1}
        \tsum_{[i:j] \in E} e_{i,j}^{(t-s)} \bP_{0,i}(s)
        \big| +
        \big|
        \tsum_{s=0}^{t-1}
        \tsum_{[i:j] \in E} e_{i,j}^{(t-s)} \bP_{0,j}(s) \big|\\
        &\leq
        \tsum_{[i:j] \in E}
        \big|
        \tsum_{s=0}^{t-1} e_{i,j}^{(t-s)}  \bP_{0,i}(s)
        \big| +
        \tsum_{[i:j] \in E}
        \big|
        \tsum_{s=0}^{t-1}
        e_{i,j}^{(t-s)} \bP_{0,j}(s)
        \big|.
    \end{align*}
    Using \thmref{cubeunimodal}, we know that the sequences $\bP_{0,i}(s)$ and $\bP_{0,j}(s)$ are unimodal in~$s$ and hence we can bound both summands by \lemref{betragkleinerk} (where $\ell=1$) to obtain that
    \begin{align}
     \big| x_0^{(t)} - \xi_0^{(t)} \big|
        &\leq 2 \Lambda \, \tsum_{[i:j] \in E}  \max_{s=0}^{t-1} \bP_{0,i}(s) +
        2 \Lambda \, \tsum_{[i:j] \in E}  \max_{s=0}^{t-1} \bP_{0,j}(s)\notag\\
        &= 2 \Lambda \, d\, \tsum_{i\in V}
         \max_{s=0}^{t-1}  \bP_{0,i}(s).  \label{eq:cubefirst}
    \end{align}

    \noindent
    To bound the last term, we conceptualize the random walk as the following process, similar to a balls-and-bins process. In
    each step $t \in \N$ we choose a coordinate $i \in \{1,\ldots,d\}$ uniformly at
    random. Then with probability~$1/2$ we flip the bit of this coordinate; otherwise
    we leave it unchanged (equivalently, we set the bit to~$1$ with
    probability~$1/2$ and to zero otherwise).

    Now we partition the random walk's distribution at step~$t$ according to the
    number of different coordinates chosen (not necessarily flipped) up to step~$t$.
    Consider $\bP_{0,x}(t)$ for a vertex $x \in \{0,1\}^d$. Note that by the symmetry of the hypercube, $\bP_{0,x}(t)$ is the same for all $x \in \{0,1\}^d$ with the same $|x|_1$. Hence, let us fix a value~$\ell$ with $0 \leq \ell \leq d$ and let us consider $\bP_{0,\ell}(t)$, which is the probability for visiting the vertex, say, $1^{\ell} 0^{d-\ell}$ from $0\equiv0^{d}$ in round $t$.
Since (i) the $k$ chosen coordinates must contain
    the $\ell$ bits which should be one and (ii) all $k$ chosen coordinates must be set to the
    correct value, we have
    \begin{equation}
       \bP_{0,\ell}(t) =
       \tsum_{k=\ell}^{d} \Pro{ \mbox{exactly $k$ coordinates chosen in $t$ steps}}
       \cdot 2^{-k}
       \, \binom{d-\ell}{k-\ell}
       \big/ \binom{d}{k}.
       \label{eq:balls}
    \end{equation}
    Using this to estimate $\bP_{0,i}(s)$, we can bound \eq{cubefirst} by
        \begin{align*}
        \big| x_0^{(t)} - \xi_0^{(t)} \big|
        &\leq 2 \Lambda \, d \, \tsum_{\ell=0}^d \binom{d}{\ell} \max_{s=0}^{\infty} \bP_{0,\ell}(s) \\
        &= 2 \Lambda \, d \, \cdot \Biggl( 1 + \tsum_{\ell=1}^d \binom{d}{\ell} \max_{s=0}^{\infty} \\
        &\quad\, \tsum_{k=\ell}^{d} \Pro{ \mbox{exactly $k$ coordinates chosen in $s$ steps}} \cdot \frac{ \binom{d-\ell}{k-\ell} }{ \binom{d}{k} } \cdot 2^{-k} \Biggr) \\
        &\leq 2 \Lambda \, d \, \cdot \Biggl( 1 +\tsum_{\ell=1}^d \max_{k=\ell}^{d} \left\{ \frac{ \binom{d-\ell}{k-\ell} \, \binom{d}{\ell} }{ \binom{d}{k} } \cdot 2^{-k} \right\} \Biggr) \\
           &= 2 \Lambda \, d \, \cdot \Biggl( 1 +\tsum_{\ell=1}^d   \max_{k=\ell}^{d} \left\{ \binom{k}{\ell} \cdot 2^{-k} \right\} \Biggr) \\
        &\leq 2 \Lambda \, d \, \cdot \Biggl( 1 +\tsum_{\ell=1}^d \max_{k=\ell}^{d} \left\{ \binom{k}{\lceil k/2 \rceil} \cdot 2^{-k} \right\} \Biggr) \\
        &\leq 2 \Lambda \, d \, \cdot \Biggl( 1 +\tsum_{\ell=1}^d \max_{k=\ell}^{d} \left\{ (1+o(1)) \cdot \frac{2^{k}}{\sqrt{\pi \, k}} \cdot 2^{-k} \right\} \Biggr) \\
        &= \Oh\left( \Lambda \, d \, \left(1 + \tsum_{\ell=1}^d \frac{1}{\sqrt{\ell}} \right) \right)
        = \Oh \big( \Lambda \, d^{3/2} \big), 
    \end{align*}
where we have used the simple consequence of Stirling's formula that $\binom{k}{\lceil \frac{k}{2} \rceil} \leq (1+o(1)) \cdot {2^k}/{\sqrt{\pi \, k}}$. This proves the first claim of the theorem.

    The second claim follows by the following simple construction.
     Define a load vector $x^{(0)}$ such that
    $x^{(0)}_v := d$ for all vertices $v=(v_1,v_2,\ldots,v_d)\in\{0,1\}^d$ with $v_1=0$, and $x^{(0)}_v := 0$ otherwise.
    Then for each edge $\{i,j\} \in E$ with $0 = i_1 \neq j_1$ the fractional flow at step $1$ is
    $
     \big(\bP_{i,j} x_i^{(0)} - \bP_{i,j} x_j^{(0)}\big) = +\frac{1}{2}.
    $
   Since in the first time step no rounding errors have yet been incurred, each edge is allowed to round up and down arbitrarily. Hence we can let all these edges round towards~$j$, i.e.,
    $\Phi_{i,j}^{(1)} := 1$ for each such edge $\{i,j\} \in E$. By definition, this implies for the corresponding rounding error that $e_{i,j}^{(1)} = -\frac{1}{2}$. Moreover, we have the following load distribution after step~$1$. We have  $x^{(1)}_v = 0$ for all vertices $v$ if $v_1=0$,
    and $x^{(1)}_v = d$ otherwise. Similarly, the fractional flow for each edge $\{i,j\} \in E$ with $0 = i_1 \neq j_1$ is
    $
     \big(\bP_{i,j} x_i^{(0)} - \bP_{i,j} x_j^{(0)}\big) = -\frac{1}{2}.
    $
    Since $e_{i,j}^{(1)} = -\frac{1}{2}$, $\big|\sum_{s=1}^2 e_{i,j}^{(s)}\big|$ will be minimized if $e_{i,j}^{(2)} = \frac{1}{2}$. Hence we can set $\Phi_{i,j}^{(2)}:=-1$. This implies that we end up in exactly the same situation as at the beginning: the load vector is the same and also the sum over the previous rounding errors along each edge is zero. We conclude that there is an instance of the BED algorithm for which $x^{(t)} = x^{(t \mod 2)}$, which proves the second claim of the \thmref{cube}.
\end{proof}

\section{Analysis on $d$-dimensional torus graphs}
\label{sec:torus}

We start this section with the formal definition of
a $d$\nobreakdash-dimensional torus. \begin{defi}
\label{def:torus}
A $d$\nobreakdash-dimensional torus $T(n_1,n_2,\ldots,n_d)$
with
$n=n_1\cdot n_2\cdot \ldots \cdot n_d$
vertices has vertex set
$V=\{0,1,\ldots,n_1-1\}\times
    \{0,1,\ldots,n_2-1\}\times
    \ldots
    \times
    \{0,1,\ldots,\linebreak[0]n_d-1\}$\linebreak[0] and every vertex
$(i_1,\linebreak[0]i_2,\linebreak[0]\ldots,\linebreak[0]i_d)\in V$ has $2d$ neighbors
$((i_1\pm 1)\bmod{n_1},\linebreak[0]i_2,\ldots,\linebreak[0]i_d)$,\linebreak[0]
$(i_1,(i_2\pm 1)\bmod{n_2},\linebreak[0]i_3,\ldots,i_d)$,
\ldots,\linebreak[0]
$(i_1,i_2,\linebreak[0]\ldots,i_{d-1},\linebreak[0](i_d\pm 1)\bmod{n_d})$. Henceforth, we will always assume that $d=\Oh(1)$.
We call a torus uniform if $n_1=n_2=\ldots=n_d=\oldsqrt[d]{n}$.
\end{defi}

Without loss of generality we will assume in the remainder that $n_1 \leq n_2
\leq \cdots \leq n_d$. By the symmetry of the torus, this does not restrict our
results.

Recall that $\lambda_2$ denotes the
second largest eigenvalue in absolute value.
Before we analyze the deviation between the idealized and discrete process, we
estimate $(1-\lambda_2)^{-1}$ for general torus graphs.
\begin{lem}\label{lem:nonuniformtorus}
    For a $d$\nobreakdash-dimensional torus $T=T(n_1,n_2,\ldots,n_d)$,
    $(1-\lambda_2)^{-1} = \Theta \left( n_d^2 \right)$.
\end{lem}
\begin{proof}
Following the notation of \cite{Ch92},
for a $k$\nobreakdash-regular graph $G$, let $\bL(G)$ be the matrix given by $\bL_{u,u}(G) = 1$,
$\bL_{u,v}(G)= -\frac{1}{k}$ if $\{u,v\} \in E(G)$ and $\bL_{u,v}(G) = 0$ otherwise.
Let $C_{q}$ be a cycle with $q$ vertices.
As shown in \cite[Example~1.4]{Ch92}, the eigenvalues of $\bL(C_{q})$
are $1 - \cos\big( \frac{2 \pi r}{q}\big)$ where $0 \leq r \leq q-1$.
In particular, the second smallest eigenvalue of $\bL(C_q)$
denoted by $\tau$ is given by $1 - \cos\big( \frac{2 \pi}{q} \big)$.

Let $\times$ denote the Cartesian product of graphs, that is, for any two graphs
$G_1=(V_1,E_1)$, $G_2=(V_2,E_2)$ the graph $G:=G_1 \times G_2$ with $G=(V,E)$
is defined by
$V=V_1 \times V_2$ and
\begin{align*}
  E := &\big\{ \big((u_1,u_2), (v_1,u_2)\big) \colon \, u_2\in V_2 \wedge \{u_1,v_1 \} \in E_1 \big\} \,\cup \\
       &\big\{ \big((u_1,u_2), (u_1,v_2)\big) \colon \, u_1\in V_1 \wedge  \{u_2,v_2 \} \in E_2 \big \}.
\end{align*}
It is straightforward to generalize this definition to the Cartesian product of
more than two graphs and it is then easy to check that $T(n_1,n_2,\ldots,n_d) =
C_{n_1} \times C_{n_2} \times \ldots \times C_{n_d}$. The following theorem
expresses the second smallest eigenvalue of the Cartesian product of graphs in
terms of the second smallest eigenvalue of the respective graphs.

\begin{thm}[{\cite[Theorem~2.12]{Ch92}}]\label{thm:cartesian}
    Let $G_1, G_2, \ldots, G_d$ be $d$ graphs and let $\tau_1, \tau_2, \ldots,
    \tau_d$ be the respective second smallest eigenvalue of $\bL(G_1), \bL(G_2), \ldots, \bL(G_d)$.
    Then the second smallest eigenvalue $\tau$ of $\bL(G_1 \times G_2 \times \ldots \times G_d)$
    satisfies
    $\tau = \tfrac{1}{d} \, \min_{k=1}^d \tau_k$.
\end{thm}

\medskip

Applying this theorem to our setting, it follows that the second
smallest eigenvalue~$\tau$ of $\bL(T)$ is
$\tau = \tfrac{1}{d} \, \big(  1 - \cos \bigl( \tfrac{2 \pi}{n_d} \bigr) \big)$.
As $n_d \geq \oldsqrt[d]{n}$, we have $ \cos \big(\frac{2 \pi}{n_d}\big)
= 1 - \Theta \big( \frac{1}{n_d^2} \big) $. Using this and the fact that $d$ is a constant, we obtain 
$
  \tau =  \Theta \big( \tfrac{1}{n_d^2} \big).
$
As $T$ is a $k$\nobreakdash-regular graph, the transition matrix~$\bP(T)$ can be expressed as
$\bP(T) = \mathbf{I} - \frac{1}{2} \bL(T)$.  This implies for the second
smallest eigenvalue of~$\bL(T)$, $\tau$, and the second largest eigenvalue of
the transition matrix~$\bP(T)$, $\lambda_2$, that
$
  \lambda_2 = 1 - \frac{1}{2} \tau.
$
Hence
$\lambda_{2} =  1 - \Theta \big( \tfrac{1}{n_d^2} \big)$,
which completes the proof.
\end{proof}

Note that the corresponding results of \citep{RSW98,FS09} only
hold for uniform torus graphs while the following result
for our algorithm holds for general torus graphs.

\begin{thm}\label{thm:torus}
    For all initial load vectors
    on the (not necessarily uniform)
    $d$\nobreakdash-dimensional torus graph with $n$~vertices,
    the deviation between the idealized process and
    a discrete process with accumulated rounding error at most~$\Lambda$
    is $\Oh(\Lambda)$ at all times.
\end{thm}

For any torus graph, we know that $(1-\lambda_2)^{-1} = \Theta( n_d^{2})$ by \lemref{nonuniformtorus}.
With \thmref{ideal} it follows that any BED algorithm (and in particular
our quasirandom algorithm) reduces the discrepancy
of any initial load vector with discrepancy $K$ to~$\Oh(1)$ within
$\Oh( n_d^{2}\,\log (K n) )$ time steps (for uniform torus graphs, this number of time steps is $\Oh(n^{2/d} \, \log (K n ) )$).

\begin{proof}[Proof of \thmref{torus}]
    By symmetry of the torus graph, we have $\bP_{i,j} = \bP_{0,i-j}$.
    Hence we set $\bP_{i} = \bP_{0,i}$.
    We will first reduce the random walk $\bP_{i,j}$
    on the finite $d$\nobreakdash-dimensional torus
    to a random walk on the infinite grid~$\Z^{d}$,
    both with loop probability~$1/2$.
    Let
    $\barP_{i,j}$ be the transition probability from~$i$ to~$j$ on~$\Z^{d}$
    defined by $\barP_{i,j} = 1/(4d)$ if $|i-j|_1=1$,
    $\barP_{i,i} = 1/2$,
    and~$0$ otherwise.
    For $i=(i_1,\ldots,i_d)\in V$
    we set
    \[
     H(i):=(i_1+n_1\,\Z,
            i_2+n_2\,\Z,
            \ldots,
            i_d+n_d\,\Z
            )\subset\Z^{d}.
    \]
    With $\barP_{i}:=\barP_{0,i}$,
    we observe that
    \[
        \bP_{i}^s = \sum_{k\in H(i)} \barP_{k}^s
    \]
    for all $s\geq0$ and $i\in V$.
    We extend the definition of $e_{i,j}$ in the natural way
    by setting
    \[
        e_{k,\ell}:=e_{i,j} \text{ for all $i,j\in V$ and $k\in H(i)$, $\ell\in H(j)$.}
    \]

    Let $\ARR=\{ \pm u_\ell \mid \ell\in\{1,\ldots,d\} \}\in\Z^{d}$ with
    $u_\ell$ being the $\ell$\nobreakdash-th unit vector.
    Following \eq{StandardAnsatz} and using the fact that
    by symmetry it suffices to bound the deviation at the vertex $0 := 0^d$,
    we get
    \begin{align*}
        x_0^{(t)} - \xi_0^{(t)}
        &=
        \dfrac{1}{2}
        \sum_{s=0}^{t-1}
        \sum_{i\in V}
        \sum_{z\in\ARR}
        e_{i,i+z}^{(t-s)}
        (\bP_{i}^s - \bP_{i+z}^s)
        \\
        &=
        \dfrac{1}{2}
        \sum_{s=0}^{t-1}
        \sum_{i\in V}
        \sum_{z\in\ARR}
        e_{i,i+z}^{(t-s)}\,
        \Bigg(\sum_{k\in H(i)} \barP_{k}^s - \sum_{\ell\in H(i+z)} \barP_{\ell}^s\Bigg)
        \\
        &=
        \dfrac{1}{2}
        \sum_{s=0}^{t-1}
        \sum_{z\in\ARR}
        \sum_{i\in V}
        e_{i,i+z}^{(t-s)}\,
        \Bigg(\sum_{k \in H(i)} \barP_{k}^s - \barP_{k+z}^s\Bigg)
        \\
        &=
        \dfrac{1}{2}
        \sum_{i\in V}
        \sum_{z\in\ARR}
        \sum_{k\in H(i)}
        \sum_{s=0}^{t-1}
        e_{k,k+z}^{(t-s)}\,
        \big( \barP_{k}^s - \barP_{k+z}^s \big)
    \end{align*}

    \noindent
    As $\Z^{d}=\bigcup_{i\in V} H(i)$ is a disjoint union, we can also write
    \begin{align}
        x_0^{(t)} - \xi_0^{(t)}
        &=
        \frac{1}{2}
        \sum_{k\in \Z^{d}}
        \sum_{z\in\ARR}
        \sum_{s=0}^{t-1}
        e_{k,k+z}^{(t-s)}\,
        \big(\barP_{k}^s - \barP_{k+z}^s\big).
        \label{eq:torus:inf}
    \end{align}

    \noindent
    We now carefully break down the sums of \eq{torus:inf} and show that each part
    can be bounded by~$\O(\Lambda)$.
    For this argument, our main tool will be \lemref{betragkleinerk}.
    As we cannot prove unimodality of $\big(\barP_{k}^s - \barP_{k+z}^s\big)$ directly,
    we will use appropriate local central limit theorems to
    approximate the transition probabilities $\barP_k^s$ of~$\Z^{d}$
    with a multivariate normal distribution.
    To derive the limiting distribution~$\tP_{k}^s$ of our random walk~$\barP_{i,j}$,
    we follow \citet{LawlerLimic} and let $X=(X_1, \ldots, X_d)$ be a $\Z^d$-valued random variable
    with $\Pr{X=z}=1/(4d)$ for all $z\in\ARR$ and
    $\Pr{X=0^d}=1/2$.
    Observe that $\Ex{X_j X_k}=0$ for $j\neq k$ since both of them cannot be non-zero
    simultaneously.
    Moreover, $\Ex{X_j X_j} = \frac{1}{4d} (-1)^2 + \frac{1}{4d} (+1)^2 = \frac{1}{2d}$
    for all
    $1\leq j\leq d$.
    Hence the covariance matrix is
    \[
        \Gamma
        := \Big[ \Ex{X_j X_k}  \Big]_{1\leq j,k\leq d}
        =(2d)^{-1} I.
    \]

    \noindent
    From Eq.~(2.2) of \citet{LawlerLimic} we get
       \begin{align}
                        \tP_{k}^s &=
            \frac{1}{(2\pi)^d\,s^{d/2}} \,
            \int_{\R^d}
            \exp \left( \im\, \frac{x \cdot k}{\sqrt{s}} \,\right) \,
            \exp \left( -\frac{x \cdot \Gamma x}{2} \right) \,
            d^d x\notag\\
        \intertext{where $\im=\sqrt{-1}$ denotes the imaginary unit.  From here, we can further conclude that}
        \tP_{k}^s
            &=
             \frac{1}{(2\pi)^d\,s^{d/2}} \,
            \int_{\R^d}
            \exp \left( \im\, \frac{x \cdot k}{\sqrt{s}}  \,
                           -\frac{x \cdot \Gamma x}{2} \right) \,
            d^d x\notag\\
            &=
             \frac{1}{(2\pi)^d\,s^{d/2}} \,
            \int_{\R^d}
            \exp \left( \im\, \frac{x \cdot k}{\sqrt{s}}  \,
                           -\frac{\|x\|_2^2}{4d} \right) \,
            d^d x\notag\\
             &=
             \frac{1}{(2\pi)^d\,s^{d/2}} \,
            \int_{\R^d}
            \exp \left( -\frac{1}{4d} \left(
                           \|x\|_2^2\,
                           - 2 \im\, \frac{2d}{\sqrt{s}} \, x \cdot k
                          \right)  \right) \,
            d^d x .\label{eq:integral1}
     \end{align}
    To evaluate the integral we complete the square, which yields
    \begin{align}
    \lefteqn{\int_{\R^d}
            \exp \left( -\frac{1}{4d} \left(
                           \|x\|_2^2\,
                           - 2 \im\, \frac{2d}{\sqrt{s}} \, x \cdot k
                          \right)  \right) \,
            d^d x}\notag \\
     &=  \int_{\R^d}\exp \left( -\frac{1}{4d} \left(
                           \|x\|_2^2\,
                           - 2 \im\, \frac{2d}{\sqrt{s}} \, x \cdot k
                           -\frac{4d^2}{s} \|k\|_2^2
                           +\frac{4d^2}{s} \|k\|_2^2
                          \right)  \right) \,
            d^d x\notag \\
     &=  \exp \left(-\frac{d}{s}\|k\|_2^2\right)
             \int_{\R^d}\exp \left( -\frac{1}{4d}
                           \left\|x-\im\,\frac{2d}{\sqrt{s}}k\right\|_2^2
                      \right) \,
            d^d x . \label{eq:integral2}
    \end{align}
    By substituting $z=x-\im\,\frac{2d}{\sqrt{s}}k$ we get
    \begin{align}
     \lefteqn{\int_{\R^d}\exp \left( -\frac{1}{4d}
                             \left\|x-\im\,\frac{2d}{\sqrt{s}}k\right\|_2^2
                            \right) \,
            d^d x
     }\notag\\
     &=
     \int_{\R^{d}}\exp \left( -\frac{1}{4d} \left(
                           \|z\|_2^2
                          \right)  \right) \,
            d^d z
            \notag\\
     &=
     \idotsint_{\R^d}\exp \left( -\frac{1}{4d} \left(
                          \sum_{i=1}^d z_i^2
                          \right)  \right) \,
            d z_d \, \ldots \, d z_1
            \notag\\
     &=
     \idotsint_{\R^{d-1}} \exp \left( -\frac{1}{4d} \left(
                          \sum_{i=1}^{d-1} z_i^2
                          \right)  \right)
                          \,
                          \left(\int_{\R}  \exp \left( -\frac{1}{4d} z_d^2\right)
            d z_d \right) d z_{d-1}\, \ldots \, d z_1
            \notag\\
     &=
     \left(2 \sqrt{\pi d}\right) \cdot
     \idotsint_{\R^{d-1}}\exp \left( -\frac{1}{4d} \left(
                          \sum_{i=1}^{d-1} z_i^2
                          \right)  \right) \,
            d z_{d-1} \, \ldots \, d z_1
            \notag\\
     &=
        \left(2 \sqrt{\pi d}\right)^{d} .
        \label{eq:integral3}
    \end{align}
    Combining \eqss{integral1}{integral2}{integral3}, we get
    \begin{align}
     \tP_{k}^s
     &=  \frac{1}{(2\pi)^d\,s^{d/2}} \,
           \exp \left(-\frac{d}{s}\|k\|_2^2\right) \,
           \left(2 \sqrt{\pi d}\right)^{d}\notag\\
     &= \left( \frac{d}{\pi s} \right) ^{d/2}
            \exp\left(\frac{-d \,\|k\|_2^2}{s}\right).
        \label{eq:multivariate}
    \end{align}
    It follows directly from Claims~4 and~5 of \citet{CooperSpencer} that
    for all $k\in\Z^{d}$, $z\in\ARR$,
    \begin{align}
        &\tP_{k}^s - \tP_{k+z}^s = \Oh( \|k\|_2^{-(d+1)}) \text{ for all $s$},
        \label{eq:multivariate:bound}
        \\
        (s\mapsto&\tP_{k}^s - \tP_{k+z}^s) \text{ has only a constant number of local extrema.}
        \label{eq:multivariate:const}
    \end{align}

    \noindent
    This gives the intuition that by approximating
    $\big(\barP_{k}^s - \barP_{k+z}^s\big)$
    with
    $\big(\tP_{k}^s - \tP_{k+z}^s\big)$,
    we can
    bound
    \eq{torus:inf}
    for sufficiently large~$k$ and~$s$
    by \lemref{betragkleinerk}.
    This approximation is made precise by the following
    local central limit theorems.
    Theorem~2.3.6 of \citet{LawlerLimic} gives
    for all $k\in\Z^{d}$, $z\in\ARR$, $s\geq0$,
    \begin{align}
        \label{eq:LCL1}
        \big|
        \big(\barP_{k}^s - \barP_{k+z}^s\big) -
        \big(\tP_{k}^s - \tP_{k+z}^s\big)
        \big|
        &=
        \O(s^{-(d+3)/2}).
    \end{align}

    \noindent
    We now start to break down \eq{torus:inf}. As the first step, we consider the special case $k$ with $\| k \|_2 < 3$, which is relatively straightforward, since it only involves a constant number of vertices. With $\Znd:=\{ k \in \Z^d \colon \| k \|_2 \geq 3 \}$
    and $\Z_{< 3}^d :=\{ k \in \Z^d \colon \| k \|_2 < 3 \}$,
        \begin{align}
        \big|
        x_0^{(t)} - \xi_0^{(t)}
        \big|
        &\leq
        \frac{1}{2}         \sum_{\substack{ k\in \Z_{< 3}^d }} 
        \overbrace{
        \Bigg|
        \sum_{z\in \ARR}
        \sum_{s=0}^{t-1}
        e_{0,0+z}^{(t-s)}\,
        \left(\barP_{k}^s - \barP_{k+z}^s\right)
        \Bigg|}^{(\ref{eq:torus:inf2}a)}
        \notag\\
        &\phantom{\leq}+
        \frac{1}{2}
        \underbrace{
        \Bigg|
        \sum_{k\in \Znd}
        \sum_{z\in \ARR}
        \sum_{s=0}^{t-1}
        e_{k,k+z}^{(t-s)}\,
        \left(\barP_{k}^s - \barP_{k+z}^s\right)
        \Bigg|
        }_{(\ref{eq:torus:inf2}b)}
        \label{eq:torus:inf2}
    \end{align}

    \noindent
    Now we can apply the local central limit theorem given in \eq{LCL1} to (\ref{eq:torus:inf2}a) and get
    \begin{align*}
        (\ref{eq:torus:inf2}a)
        &=
        \Bigg|
        \sum_{z\in \ARR}
        \sum_{s=0}^{t-1}
        e_{k,k+z}^{(t-s)}\,
        \big(\tP_{0}^s - \tP_{0+z}^s\big)
        \Bigg|
        +
        \Bigg|
        \sum_{z\in \ARR}
        \sum_{s=0}^{t-1}
        \O(s^{-(d+3)/2})
        \Bigg|
        =
        \O(\Lambda),
    \end{align*}
    where the last equality follows by
    \lemref{betragkleinerk}
    combined with
    \eq{multivariate:const}
    and
    the property
    $\big|\sum_{s=1}^t e_{i,j}^{(s)} \big|\leq \Lambda$.
    
    In order to analyze the vertices $k \in\Znd$ in \eq{torus:inf}, 
    we proceed by fixing
    a cutoff point $T(k):=\tfrac{C \,\|k\|_2^2}{\ln^2(\|k\|_2)}$,
    $k\in\Znd$, of the innermost sum
    of (\ref{eq:torus:inf2}b)
    for some sufficiently small constant $C>0$. Note that the cutoff point is chosen in such a way that it is very unlikely for a random walk to reach a vertex with distance $k$ within less than $T(k)$ rounds. Hence the contribution from all those summands (\ref{eq:torus:inf3}a) is only a constant. For bounding 
the remaining sum, (\ref{eq:torus:inf3}b), we 
use the local central limit theorem (\ref{eq:LCL1}) whose approximation error is small, since we only need to consider rounds larger than the cutoff $T(k)$.
   
    Proceeding with the formal proof, we have 
    \begin{align}
        (\ref{eq:torus:inf2}b)
        &\leq
        \overbrace{
        \Bigg|
        \sum_{k\in \Znd}
        \sum_{z\in \ARR}
        \sum_{s=0}^{T(k)}
        e_{k,k+z}^{(t-s)}\,
        \left(\barP_{k}^s - \barP_{k+z}^s\right)
        \Bigg|}^{(\ref{eq:torus:inf3}a)}
        \notag\\
        &\phantom{\leq}+
        \underbrace{
        \Bigg|
        \sum_{k\in \Znd}
        \sum_{z\in \ARR}
        \sum_{s=T(k)}^{t-1}
        e_{k,k+z}^{(t-s)}\,
        \left(\barP_{k}^s - \barP_{k+z}^s\right)
        \Bigg|}_{(\ref{eq:torus:inf3}b)}.
        \label{eq:torus:inf3}
    \end{align}
    Note that the summand with~$s=0$ is zero and can be ignored (since $\|k\|_2 \geq 3$).
    Hence the first summand (\ref{eq:torus:inf3}a) can be bounded by
    \begin{align}
        (\ref{eq:torus:inf3}a)
        &=
        \O
        \Bigg(
        \sum_{k\in \Znd}
        \sum_{z\in \ARR}
        \sum_{s=1}^{T(k)}
        \Big(\barP_{k}^s + \barP_{k+z}^s\Big) \label{eq:secondlast}
        \Bigg).
    \end{align}
    It is known from
    \citet[Lem.~1.5.1(a)]{Lawler} that
    for random walks on infinite grids,
    $\sum_{\|k\|_2\geq\lambda\sqrt{s}} \barP_{k}^s = \O(e^{-\lambda})$
    for all $s>0$ and $\lambda>0$. Hence, it is also true that
        \[
    \barP_{k}^s
    = \O\big(\exp\big(-\|k\|_2 / \sqrt{s} \big)\big)
    \text{ for all $s>0$, $k\in\Z^{d}$.}
    \]
                    With that, we can now bound the term $\big(\barP_{k}^s + \barP_{k+z}^s\big)$ from \eq{secondlast}.
    For $0<s\leq T(k)$,
    $k\in\Znd$, $z\in\ARR$,
    and sufficiently small $C > 0$,
           \begin{align*}
    \barP_{k}^s + \barP_{k+z}^s
        &=\Oh\bigg( \exp\bigg(-\frac{\|k\|_2}{\sqrt{s}}\bigg) + \exp\bigg(-\frac{\|k + z\|_2}{\sqrt{s}}\bigg) \bigg) \\
    &=\Oh\bigg( \exp\bigg(- \frac{ \| k \|_2}{\sqrt{3} \cdot \sqrt{s}}\bigg) \bigg) \\
    &=\Oh\bigg(\exp\bigg(- 2 \ln(\|k\|_2)\frac{\|k\|_2}{\sqrt{3 \cdot C}\, \|k\|_2}\bigg)\bigg) \\
    &=\Oh\big(\|k\|_2^{-(d+4)}\big),
    \end{align*}
    where in the second line we have used 
    that
    \[
       \| k \|_2^2 - \| k + z \|_2^2  \leq \| k \|_2^2 - \left( \|k\|_2^2 - 2 \|k\|_2 \cdot \|z\|_2 + \|z\|_2^2  \right) \leq 2 \| k \|_2 ,
     \]
    implying directly that \[
    \|k+z\|_2 \geq \sqrt{ \|k \|_2^2 - 2 \|k \|_2 } = \sqrt{ \|k\|_2 \cdot ( \|k\|_2 - 2  )} \geq  \frac{ \|k\|_2 }{\sqrt{3}}, \] since $\|k\|_2 \geq 3$.

    \noindent
    Plugging this into \eq{secondlast}, we obtain
    \begin{align*}
        (\ref{eq:torus:inf3}a)
        =
        \O\Bigg(
        \sum_{k\in \Znd}
        T(k)\,
        \|k\|_2^{-(d+4)}
        \Bigg)
        =
        \O\Bigg(
        \sum_{k\in \Znd}
        \|k\|_2^{-(d+2)}
        \,\ln^{-2}(\|k\|_2)
        \Bigg)
        =\O(1).
    \end{align*}

    To bound (\ref{eq:torus:inf3}b),
    we approximate the transition probabilities of $\Z^{d}$
    with the multivariate normal distribution of \eq{multivariate}
    by the local central limit theorem stated in \eq{LCL1}:
    \begin{align}
        (\ref{eq:torus:inf3}b) &= \Bigg|
        \sum_{k\in \Znd}
        \sum_{z\in \ARR}
        \sum_{s=T(k)}^{t-1}
        e_{k,k+z}^{(t-s)}\,
        \left(\tP_{k}^s - \tP_{k+z}^s\right)
        \notag \\
        &\phantom{\leq}+
        \sum_{k\in \Znd}
        \sum_{z\in \ARR}
        \sum_{s=T(k)}^{t-1}
        e_{k,k+z}^{(t-s)}\,
        \left(\barP_{k}^s - \barP_{k+z}^s\right) - \left(\tP_{k}^s - \tP_{k+z}^s\right)
        \Bigg|   \notag \\
        &\leq
        \overbrace{
        \Bigg|
        \sum_{k\in \Znd}
        \sum_{z\in \ARR}
        \sum_{s=T(k)}^{t-1}
        e_{k,k+z}^{(t-s)}\,
        \left(\tP_{k}^s - \tP_{k+z}^s\right)
        \Bigg|}^{(\ref{eq:megasum}a)}
        \notag\\
        &\phantom{\leq}+
        \underbrace{
        \Bigg|
        \sum_{k\in \Znd}
        \sum_{z\in \ARR}
        \sum_{s=T(k)}^{t-1}
        e_{k,k+z}^{(t-s)}\,
         \O(s^{-(d+3)/2})
        \Bigg|}_{(\ref{eq:megasum}b)}.
        \label{eq:megasum}
    \end{align}

    \noindent
    We can bound the second term (\ref{eq:megasum}b), signifying the approximation error from the local central limit theorem, by
    \begin{align*}
        (\ref{eq:megasum}b)
        &=
        \O\Bigg(
        d
        \sum_{k\in \Znd}
        \sum_{s=T(k)}^{\infty}
        \,s^{-(d+3)/2}
        \Bigg)
        \\
        &=
        \O\Bigg(
        \sum_{k\in \Znd}
        T(k)^{-(d+1)/2}
        \Bigg) \\
         &=
        \O\Bigg(
        \sum_{k\in \Znd}
        \frac{\ln^{d+1}(\|k\|_2)  }{\|k\|_2^{d+1}}
        \Bigg).
    \intertext{As there are constants $C' > 0$ and $\epsilon > 0$ such that
    $\ln^{d+1}(\|k\|_2) \leq C' \, \|k\|_2^{1-\epsilon}$ for all $k \in \Znd$, we obtain}
        (\ref{eq:megasum}b)
        &=
        \O\Bigg(
        \sum_{k\in \Znd}
        \|k\|_2^{-(d+\epsilon)}
        \Bigg).
    \end{align*}
    To see that this can be bounded by $\O(1)$, observe that
    with $\Nnd:= \{ k \in \N^d \colon \|k\|_2 \geq 3 \}$,
    \begin{align*}
      \sum_{k\in \Znd} \| k \|_2^{-(d+\epsilon)}
      \leq 2^{d} \, \sum_{k\in \Nnd} (k_1^2 + \cdots + k_{d}^2)^{-(d+\epsilon)/2}.
    \end{align*}
    By convexity of $x \mapsto x^2$,
    $k_1^{2} + \cdots + k_{d}^{2}
    \geq \tfrac{1}{d} (k_1 + \cdots + k_{d})^2$, we then get
    \begin{align*}
       (\ref{eq:megasum}b)
       &= \Oh \Bigg( \sum_{k\in \Nnd} (k_1 + \cdots + k_{d})^{-(d+\epsilon)}\Bigg)
       = \Oh \Bigg( \sum_{x=1}^{\infty} \sum_{\substack{k\in \N^{d} \\ \|k\|_1=x}}
        x^{-(d+\epsilon)} \Bigg) \\
      &= \Oh \Bigg( \sum_{x=1}^{\infty} x^{d-1} \cdot x^{-(d+\epsilon)} \Bigg)
       = \Oh \Bigg( \sum_{x=1}^{\infty} x^{-(1+\epsilon)} \Bigg)
       = \Oh(1).
    \end{align*}

    \noindent
    Finally, to bound (\ref{eq:megasum}a),
    we apply
    \eq{multivariate:const}.
    We also use the fact that $\tP_{k}^s - \tP_{k+z}^s$
    can be bounded by $\Oh( \|k\|_2^{-(d+1)})$
    according to \eq{multivariate:bound}.
    As     $|\sum_{s=1}^t e_{i,j}^{(s)} |\leq \Lambda$,
    applying \lemref{betragkleinerk} yields
    \begin{align*}
        (\ref{eq:megasum}a)
        &=
        \O\Bigg(
        \sum_{k\in \Znd}
        \sum_{z\in \ARR}
        \Lambda \,
        \max_{s=T(k)}^{t-1}
        \Big(\tP_{k}^s - \tP_{k+z}^s\Big)
        \Bigg)\\
        &=\O\Bigg(
        \sum_{k\in \Znd}
        \sum_{z\in \ARR}
        \Lambda \,
        \|k\|_2^{-(d+1)}
        \Bigg)
        \\&
        =
        \O\Bigg( \Lambda \,
        d
        \sum_{k\in \Znd}
        \|k\|_2^{-(d+1)}
        \Bigg)
        \\&
        =
        \O ( \Lambda ).
    \end{align*}
    Combining all the above bounds,
    we can conclude that
    $
      \big|x_0^{(t)} - \xi_0^{(t)}\big|
        = \Oh(\Lambda),
    $
    meaning that the deviation between the idealized process and the discrete process at any time and
    vertex is at most $\Oh(\Lambda)$.
\end{proof}

\section{Lower bounds for previous algorithms}
\label{sec:others}

For a better comparison with previous algorithms,
this section gives lower bounds for other discrete diffusion processes.
First, we observe the following general lower bound on the discrepancy for the RSW algorithm.

\begin{pro}
\label{pro:lowerbound}
    On all
    graphs~$G$ with maximum degree~$\Delta$,
    there is an initial load-vector $x^{(0)}$ with discrepancy $\Delta \diam(G)$
    such that
    for the RSW algorithm,
    $x^{(t)}=x^{(t-1)}$ for all $t \in \N$.
\end{pro}
\begin{proof}
    Fix a pair of vertices $i$ and $j$ with $\dist(i,j)=\diam(G)$. Define an initial load-vector $x^{(0)}$ by
    \[
      x^{(0)}_k := \dist(k,i) \cdot \Delta.
    \]
    Clearly, the initial discrepancy is $x_j^{(0)}-x_i^{(0)} = \Delta \diam(G)$.
    We claim that $x^{(1)}=x^{(0)}$. Consider an arbitrary edge $\{r,s\} \in E(G)$. Then,
    \begin{align*}
       \big| \bP_{r,s}\,x_r^{(1)} - \bP_{s,r} \,x_s^{(1)}  \big|
       = \frac{1}{2 \Delta} \big| x_r^{(0)}  - x_s^{(0)} \big|
       \leq \frac{1}{2 \Delta}  \, \Delta
       = \frac{1}{2}.
    \end{align*}
    Hence the integral flow on any edge $\{r,s\} \in E(G)$ is $\lfloor \frac{1}{2} \rfloor = 0$ and the load-vector remains unchanged. The claim follows.
\end{proof}

In the remainder of this section we present two lower bounds for
the deviation between the randomized rounding diffusion
algorithm and the idealized process.
First, we prove a bound of $\Omega(\log n)$
for the hypercube.  Together with \thmref{cube}
this implies that on hypercubes
the quasirandom approach is as good as the
randomized rounding diffusion algorithm.

\begin{thm}
\label{thm:cuberandomlower}
    There is an initial load vector
    of the $d$\nobreakdash-dimensional hypercube with $n=2^d$ vertices
    such that the deviation of the randomized rounding diffusion
    algorithm and the idealized process
    is at least $\log n/4$ with probability $1-n^{\Omega(1)}$.
\end{thm}
\begin{proof}
    We define an initial load vector $x^{(0)}$ as follows. For every vertex
    $v=(v_1,v_2,\ldots,v_{d}) \in \{0,1\}^{d}$ with $v_1=0$ we set
     $x_{v}^{(0)}=\xi_{v}^{(0)}=0$ and if $v_1=1$
     we set $x_{v}^{(0)}=\xi_{v}^{(0)}=d$.
     Hence, the idealized process
    will send a flow of $d/(2d)=1/2$ from
    every vertex $v=(1,v_2,v_3,\ldots,v_d) \in \{0,1\}^d$ to $(0,v_2,v_3,\ldots,v_d)$. Hence for the idealized process, $\xi_v^{(1)}=(1/2)\,d$, that is, all
    vertices have a load of $(1/2)\,d$ after one step
    and the load is perfectly balanced.

    Let us now consider the discrete
    process. Let $V_0$ be the set of vertices whose bitstring begins with $0$.
    Consider any node $v \in V_0$. Note that all neighbors of $v$ have a load of $1$ and the integral flow from any of those neighbors equals $1$ with independent probability $1/2$. Hence the load of $v$ in the next step is just a binomial random variable, and
     using the fact that $\binom{r}{s} \geq (r/s)^s$, we obtain
    \begin{align*}
      \Pro { x_{v}^{(1)} = \frac{3}{4} \, d } \geq \Pro { x_{v}^{(1)} \geq \frac{3}{4} \, d } &\geq \binom{d}{(3/4) d} 2^{-d} \geq \bigg( \frac{4}{3} \bigg)^{(3/4) d} 2^{-d}  \geq n^{-1+C},
     \end{align*}
    for some constant $C > 0$ since $d=\log_2 n$. As the maximum degree of the graph is $\log n$ and the size of $V_0$ is $n/2$, it follows that there is a subset $S
    \subseteq V_0$ of size $\Omega\big(\frac{n}{\log^4 n}\big)$ in the hypercube such that
    every pair in $S$ has distance at least $4$. By construction,
    the respective events $x_{v}^{(1)} \geq (3/4) d$ are independent for all vertices $v \in S$. Hence
    \begin{align*}
      \Pro { \exists v \in S \colon x_{v}^{(1)} \geq \frac{3}{4} \, d } &\geq 1 - \left( 1 - n^{-1+C} \right)^{\Omega\left(\frac{n}{\log^4 n}\right)} \geq 1 - n^{-C'},
    \end{align*}
  where $0 < C' < C$ is another constant. This means that with probability at least $1-n^{-C'}$ the load at vertex $u$ at step $1$ will be at least $(3/4)\,d$ in the discrete process, but equals $(1/2)\,d$ in the idealized process. This completes the proof.
\end{proof}

It remains to give a lower bound for the deviation
between the randomized rounding and the idealized process
for torus graphs.
The following theorem proves a polylogarithmic lower bound
for the randomized rounding algorithm, which
should be compared to the constant upper bound
for the quasirandom approach of \thmref{torus}.
Similar results can also be derived for non-uniform torus graphs.

\begin{thm}
\label{thm:toruslower}
    There is an initial load vector
    of the $d$\nobreakdash-dimensional uniform torus graph with $n$~vertices
    such that the deviation between the randomized rounding diffusion algorithm
    and the idealized process
    is $\Omega( \polylog(n))$ with probability $1-o(1)$.
\end{thm}

\newcommand{\T}{\mathsf{T}}
\begin{proof}
    Let $n$ be a sufficiently large integer and
    $\T$ be a $d$\nobreakdash-dimensional torus graph with $n$~vertices
    and side length $\oldsqrt[d]{n}\in\N$.
    Let $B_{k}(u) := \linebreak[0]\{ v \in V \colon \|v-u\|_\infty \leq k \}$ and
    $\partial B_{k}(u) := \{ v \in V \colon \|v-u\|_\infty = k \}$.
    For every vertex $v \in V(\T)$, we define
    $|B_{\ell/2}(v)| = \ell^d = (\log n)^{1/4} $
    with $\ell:=(\log n)^{1/(4d)}$, where we assume w.l.o.g. that $\ell$ is an odd integer.
    For $\ell':=(\log n)^{2/(3d)}$, define a set $S \subseteq V$ by
    \[
		S := \bigl \{
		( x_1 \, \ell', x_2 \, \ell', \ldots, x_{d} \, \ell')
		\, \bigm\vert \,
        1 \leq x_1,x_2, \ldots,x_d < \oldsqrt[d]{n}/\ell' - 1
        \bigr\},
	\]
	that is, every pair of distinct vertices in $S$ has a coordinate-wise
	distance which is a multiple of $\ell'$.
	Note that $|S|=\Omega(n/\ell'^{d})$.
    Define the initial load vector as
    $x_i^{(0)} = \xi_i^{(0)} := 2d \cdot \max\{0, \ell/2 - \dist(i,S) \}$,
    $i\in V$.
    Clearly, the initial discrepancy is $K=2d \cdot \ell/2$.

    The idea is now to decompose $\T$ into smaller subgraphs
    centered around $s\in S$, since
    the upper bound on the convergence rate given by
    \thmref{ideal} has a strong dependence on the size of the graph.
    Then we relate the simultaneous convergence on each of the
    smaller graphs to the convergence on the original graph. An illustration of
    our decomposition of $\T$ can be found in \figref{torus}.

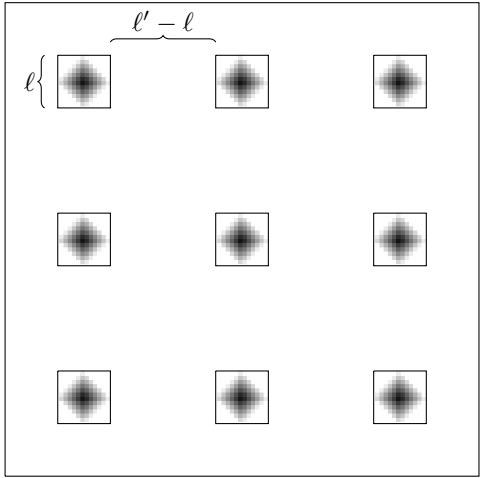
\begin{figure}[tb]
\center
\begin{tikzpicture}[auto,domain=0:4,x=1cm,y=1cm]
\pgftransformxscale{.7}
\pgftransformyscale{.7}

\draw (0,1) -- (9,1) -- (9,10) -- (0,10) -- (0,1);

\draw[snake=brace] (0.75,8) -- (0.75,9);
\draw[snake=brace] (2,9.25) -- (4,9.25);
\draw(0.75,8.5) node[left] {$\ell$};
\draw(3,9.25) node[above] {$\ell' - \ell$};

\setcounter{xcounter}{0}
\foreach \x  in {1.08,1.16,...,2}
\foreach \y  in {2.08,2.16,...,3}
{
\pgfmathsetcounter{cntShader}{100-190*max(\x-1.5,1.5-\x)-190*max(\y-2.5,2.5-\y)}
\filldraw [fill=\couleur!\thecntShader, color=\couleur!\thecntShader, line width=0pt] (\x-0.04,\y-0.04) rectangle (\x+0.04,\y+0.04);
}

\setcounter{xcounter}{0}
\foreach \x  in {1.08,1.16,...,2}
\foreach \y  in {5.08,5.16,...,6}
{
\pgfmathsetcounter{cntShader}{100-190*max(\x-1.5,1.5-\x)-190*max(\y-5.5,5.5-\y)}
\filldraw [fill=\couleur!\thecntShader, color=\couleur!\thecntShader, line width=0pt] (\x-0.04,\y-0.04) rectangle (\x+0.04,\y+0.04);
}

\setcounter{xcounter}{0}
\foreach \x  in {1.08,1.16,...,2}
\foreach \y  in {8.08,8.16,...,9}
{
\pgfmathsetcounter{cntShader}{100-190*max(\x-1.5,1.5-\x)-190*max(\y-8.5,8.5-\y)}
\filldraw [fill=\couleur!\thecntShader, color=\couleur!\thecntShader, line width=0pt] (\x-0.04,\y-0.04) rectangle (\x+0.04,\y+0.04);
}

\setcounter{xcounter}{0}
\foreach \x  in {4.08,4.16,...,5}
\foreach \y  in {2.08,2.16,...,3}
{
\pgfmathsetcounter{cntShader}{100-190*max(\x-4.5,4.5-\x)-190*max(\y-2.5,2.5-\y)}
\filldraw [fill=\couleur!\thecntShader, color=\couleur!\thecntShader, line width=0pt] (\x-0.04,\y-0.04) rectangle (\x+0.04,\y+0.04);
}

\setcounter{xcounter}{0}
\foreach \x  in {4.08,4.16,...,5}
\foreach \y  in {5.08,5.16,...,6}
{
\pgfmathsetcounter{cntShader}{100-190*max(\x-4.5,4.5-\x)-190*max(\y-5.5,5.5-\y)}
\filldraw [fill=\couleur!\thecntShader, color=\couleur!\thecntShader, line width=0pt] (\x-0.04,\y-0.04) rectangle (\x+0.04,\y+0.04);
}

\setcounter{xcounter}{0}
\foreach \x  in {4.08,4.16,...,5}
\foreach \y  in {8.08,8.16,...,9}
{
\pgfmathsetcounter{cntShader}{100-190*max(\x-4.5,4.5-\x)-190*max(\y-8.5,8.5-\y)}
\filldraw [fill=\couleur!\thecntShader, color=\couleur!\thecntShader, line width=0pt] (\x-0.04,\y-0.04) rectangle (\x+0.04,\y+0.04);
}

\setcounter{xcounter}{0}
\foreach \x  in {7.08,7.16,...,8}
\foreach \y  in {2.08,2.16,...,3}
{
\pgfmathsetcounter{cntShader}{100-190*max(\x-7.5,7.5-\x)-190*max(\y-2.5,2.5-\y)}
\filldraw [fill=\couleur!\thecntShader, color=\couleur!\thecntShader, line width=0pt] (\x-0.04,\y-0.04) rectangle (\x+0.04,\y+0.04);
}

\setcounter{xcounter}{0}
\foreach \x  in {7.08,7.16,...,8}
\foreach \y  in {5.08,5.16,...,6}
{
\pgfmathsetcounter{cntShader}{100-190*max(\x-7.5,7.5-\x)-190*max(\y-5.5,5.5-\y)}
\filldraw [fill=\couleur!\thecntShader, color=\couleur!\thecntShader, line width=0pt] (\x-0.04,\y-0.04) rectangle (\x+0.04,\y+0.04);
}

\setcounter{xcounter}{0}
\foreach \x  in {7.08,7.16,...,8}
\foreach \y  in {8.08,8.16,...,9}
{
\pgfmathsetcounter{cntShader}{100-190*max(\x-7.5,7.5-\x)-190*max(\y-8.5,8.5-\y)}
\filldraw [fill=\couleur!\thecntShader, color=\couleur!\thecntShader, line width=0pt] (\x-0.04,\y-0.04) rectangle (\x+0.04,\y+0.04);
}

\draw (1,9) -- (2,9) -- (2,8) -- (1,8) -- (1,9);
\draw (4,9) -- (5,9) -- (5,8) -- (4,8) -- (4,9);
\draw (1,6) -- (2,6) -- (2,5) -- (1,5) -- (1,6);
\draw (4,6) -- (5,6) -- (5,5) -- (4,5) -- (4,6);
\draw  (7,9) -- (8,9) -- (8,8) -- (7,8) -- (7,9);
\draw  (7,6) -- (8,6) -- (8,5) -- (7,5) -- (7,6);
\draw  (1,3) -- (2,3) -- (2,2) -- (1,2) -- (1,3);
\draw  (4,3) -- (5,3) -- (5,2) -- (4,2) -- (4,3);
\draw  (7,3) -- (8,3) -- (8,2) -- (7,2) -- (7,3);

\end{tikzpicture}
\caption{Overview of the decomposition of $\T$ into various $\T'(s)$ for the two-dimensional case $d=2$. The inner rectangles represent the various smaller grids $\T'(s)$ with $s \in S$. The darkness indicates the amount of the initial load. Note that the initial load of vertices outside the $\T'(s)$'s is~$0$.} \label{fig:torus}
\end{figure}

    Fix some $s\in S$.  Then the subgraph
    induced by the vertices $B_{\ell/2}(s)$
    is a $d$\nobreakdash-dimensional grid with exactly $n':=(\log n)^{1/4}$
    vertices.
    Let $\T'=\T'(s)$ denote the corresponding $d$\nobreakdash-dimensional torus graph with
    the same vertices, but additional
    wrap-around edges between vertices of
    $\partial B_{\ell/2}(s)$.
    W.l.o.g.\ we assume that the side length $\oldsqrt[d]{n}$ of $T$ is a multiple
    of the side length $\ell$
    of $\T'(s)$.
    Let $\bP'$ be the diffusion matrix of $\T'(s)$.

    Let us denote by $\xi'^{(0)}$ $(x'^{(0)})$ the projection
    of the load vector $\xi^{(0)}$ $(x^{(0)})$ from $\T$ onto $\T'(s)$.
    By \corref{ideal}, the idealized process reduces the discrepancy on $\T'(s)$
    from~$K=(\log n)^{1/(4d)}/2$ to~$1$ within
    $t_0 := \Oh( (n')^{2/d}\,\log (K n') )
    = \Oh( \log \log (n) \, (\log n)^{1/(2d)})$ time steps.
    We now want to argue that this also happens on the original graph $\T$ with $n$~vertices.
    Note that the convergence of the idealized process on $\T'(s)$ implies
    \begin{equation}
      \|\xi'^{(t_0)} - \overline{ \xi'} \|_{\infty}
      = \| {\bP'}^{t_0} \xi'^{(0)}  -\overline{ \xi'} \|_{\infty} \leq 1. \label{eq:firstt}
    \end{equation}
    Furthermore, note that the average load $\overline{\xi'}$ in each $\T'(s)$ satisfies
    \[
      \overline{\xi'} \leq 2d \cdot \ell/4.
    \]

    \noindent
    Our next observation is that for any two vertices $u,v \in \T'(s)$,
    \begin{align}
       \bP_{u,v}^{t_0} &\leq {\bP'}_{u,v}^{t_0} \label{eq:secondd}
    \end{align}
    as a random walk on $\T'(s)$ can be expressed as a projection of a random
    walk on $\T$ (by assigning each vertex in $\T'(s)$ to a set of vertices
    in $\T$). With the observations
    \begin{itemize}
    \setlength{\itemsep}{0pt}
    \setlength{\parskip}{0pt}
       \item for $v\in \T'(s)$:  ${\xi}_v^{(0)}={\xi'_v}^{(0)}$,
       \item for $v \in \T$ and $\ell/2 \leq \dist(v,s) \leq t_0$:
             $\xi_v^{(0)}=0$
             (as $t_0=o(\ell'-\ell/2))$,
       \item for $v\in \T$ and $\dist(v,s)>t_0$:  $\bP_{v,s}^{t_0}=0$,
    \end{itemize}
    we obtain for any vertex $v \in  B_{\ell/2}(s)$,
    \[
         \xi_{s}^{(t_0)}
         = \big( \bP^{(t_0)} \cdot \xi^{(0)} \big)_s
         = \sum_{v \in \T} \xi_v^{(0)} \bP_{v,s}^{(t_0)}
         = \sum_{v \in \T'(s)} {\xi'_v}^{(0)} \bP_{v,s}^{(t_0)}.
    \]
    By first applying \eq{secondd} and then \eq{firstt}, we get
    \[
     \xi_{s}^{(t_0)}
     \leq  \sum_{v \in \T'(s)} {\xi'_s}^{(0)} {\bP'}_{v,s}^{(t_0)}
     = {\xi'_s}^{(t_0)} \leq \overline{ \xi'} + 1.
    \]

    \noindent
    This means that after $t_0$ time steps, the idealized process achieves
    a good balancing at $s$.
    On the other hand, the discrete process may fail within $t_0$ time steps if
    there is an~$s$ such that all edges in $\T'(s)$ round towards $s$ at all
    time steps $t\leq t_0$. (Note that by construction, no load from another
    $\T'(s')$, $s' \in S \setminus \{s\}$, can reach $\T'(s)$ within $t_0$ steps,
    since the distance between any vertex in $\T'(s)$ and $\T'(s')$ is $\ell' -
    2 \ell \geq t_0$.)
    Moreover, by definition of $x^{(0)}$, $|x^{(0)}_u -
    x^{(0)}_v|\in \{0, 2d \}$   if $\{u,v\}\in E(\T)$. Hence the fractional flow in
    the first step is $\in \{0, \frac{1}{2} \}$
    and for fixed $s$
    the probability that $x^{(0)}_u=x^{(1)}_u$ for all $u\in \T'(s)$ is
    at least
    $2^{ -|B_{\ell/2}(s)| }$.
    By induction, for fixed $s$ the probability that $x^{(0)}_u=x^{(t_0)}_u$
    holds for all $u\in \T'(s)$ is at least
    \[
        2^{ -|B_{\ell/2}(s)| \, t_0}
        = 2^{ -(\log n)^{1/4} \cdot \Oh( \log \log (n) \, (\log n)^{1/(2d)} )} \geq 2^{-(\log n)^{4/5}} .
    \]
    As we have $|S| = \Omega(n/\ell'^{d})=\Omega(\poly(n))$
    independent events, it follows that
    there is at least one $s \in S$ with $ x_{s}^{(t_0)} = x_s^{(0)} = \ell/2 \cdot 2d$
    with probability
    \[
         1 - \Big(1 - 2^{-(\log n)^{4/5}} \Big)^{\Omega(\poly{(n)})} \geq 1- n^{-C},
    \] where $C > 0$ is some constant.
    If this happens, then the deviation between the discrete and idealized process at vertex $s \in S$ at step $t_0$ is at least
    \[
     \big|   x_{s}^{(t_0)} - \xi_{s}^{(t_0)} \big|
     \geq \big| 2d \cdot \ell/2 - (2d \cdot \ell/4 + 1) \big|
     = \Omega( (\log n)^{1/(4d)}),
    \]
 and the claim follows.
\end{proof}

\section{Conclusions}

We propose and analyze a new deterministic algorithm for balancing
indivisible tokens.  By achieving a constant discrepancy in optimal time
on all torus graphs, our algorithm improves upon all previous deterministic and random
approaches with respect to both running time and discrepancy.
For hypercubes we prove a discrepancy of $\Theta(\log n)$, which
is also significantly better than the (deterministic)
RSW algorithm, which achieves a discrepancy of $\Omega(\log^2 n)$.

On a concrete level, it would be interesting to extend these results
to other network topologies. From a higher-level perspective, our new algorithm provides
a striking example of quasirandomness in algorithmics. Devising and analyzing
similar algorithms for other tasks such as routing, scheduling or synchronization
remains an interesting open problem.

\appendix

\section{Proof of a supplementary result}

\noindent
In order to prove \proref{monotonicity}, we first note the following elementary lemma.
\begin{lem}
\label{lem:basiclemma2}
    Let $(a_k)_{k=1}^{d}, (b_k)_{k=1}^{d}, (c_k)_{k=1}^{d}$ be three positive sequences such that
    \begin{enumerate}
        \item for all $j \in [1,d]$, $\sum_{k=j}^d a_k \leq \sum_{k=j}^d b_k$,
        \item $c_k$ is monotone increasing in $k$.
    \end{enumerate}
    Then for all $j \in [1,d]$,
    $
      \sum_{k=j}^d a_k \cdot c_k \leq \sum_{k=j}^d b_k \cdot c_k.
    $
\end{lem}
\begin{proof}
Define $i:=d-j+1$. We will show that
\begin{align}
   \sum_{k=d-i+1}^{d} a_k \cdot c_k &\leq \sum_{k=d-i+1}^{d} b_k \cdot c_k
   \label{e:indhypoapp}
\end{align}
for all $i\in[1,d]$.
Our proof is by induction on the number of summands $i \in [1,d]$. The claim is trivial
for $i=1$ and $i=2$.
Assume inductively that \eqref{e:indhypoapp} holds for all $i\ge i'$ and for all sequences satisfying the
conditions of the lemma. We will show the claim for $i=i'+1$, i.e.,
\begin{align*}
   \sum_{k=d-i'}^{d} a_k \cdot c_k &\leq \sum_{k=d-i'}^{d} b_k\cdot c_k.
\end{align*}
Define two shorter sequences $(a_k')_{k=1}^{d-1}$ and $(b_k')_{k=1}^{d-1}$
 as follows:
\begin{itemize}
\item
 $a'_{k} = a_{k}$ for $k < d-1$ and  $a'_{d-1} := a_{d-1} + \frac{c_d}{c_{d-1}} a_{d}$
\item
$b'_{k} = b_{k}$ for $k < d-1$ and  $b'_{d-1} := b_{d-1} + \frac{c_d}{c_{d-1}} b_{d}$
\end{itemize}
We will show that the sequences $(a'_k)_{k=1}^{d-1}, (b'_k)_{k=1}^{d-1}, (c_k)_{k=1}^{d-1}$
satisfy the conditions of the lemma. Since $(c_k)_k$ remained unchanged it suffices to show
that
 for all $j' \in [1,d-1]$,
 \[
    \sum_{k=j'}^{d-1} a'_k \leq \sum_{k=j'}^{d-1} b'_k,
 \]
or equivalently (using the definition of $a'_k$ and $b'_k$)
 \begin{align*}
    \sum_{k=j'}^{d} a_k + \left(\frac{c_d}{c_{d-1}}-1\right) a_{d}
    &\leq \sum_{k=j'}^{d} b_k + \left(\frac{c_d}{c_{d-1}}-1\right) b_{d}.
 \end{align*}
 By the first assumption of the lemma, we have
  \[
   \sum_{k=j'}^{d} a_k \leq \sum_{k=j'}^{d} b_k.
 \]
 Moreover, since $a_d \leq b_d$ and
 $\frac{c_d}{c_{d-1}} \geq 1$, we have
 \[
    \left(\frac{c_d}{c_{d-1}}-1\right) a_{d} \le \left(\frac{c_d}{c_{d-1}}-1\right) b_{d}.
 \]
 Thus,  $(a'_k)_{k=1}^{d-1}, (b'_k)_{k=1}^{d-1}, (c_k)_{k=1}^{d-1}$ satisfy the conditions of the
 lemma.
 By the induction hypothesis
 on those sequences and for $i'$ summands,  we have
\begin{align*}
   \sum_{k=d-i'}^{d-1} a'_k \cdot c_k &\leq \sum_{k=d-i'}^{d-1} b'_k\cdot c_k.
\end{align*}
Plugging in the definition of $a'_k$ and $b'_k$, we finally obtain
 \begin{align*}
  \sum_{k=d-i'}^{d-2} a_k \cdot c_k + c_{d-1} \, \left( a_{d-1} + \frac{c_d}{c_{d-1}} a_{d}  \right)
  &\leq
  \sum_{k=d-i'}^{d-2} b_k \cdot c_k + c_{d-1} \, \left( b_{d-1} + \frac{c_d}{c_{d-1}} b_{d}  \right),
 \end{align*}
 which is precisely the induction claim for $i'+1$. The lemma follows.
\end{proof}

\againpronofrom{monotonicity}{\textpromonotonicity}
\begin{proof}
    Fix an arbitrary step $t \in \N$.
    By symmetry, it suffices to prove that $\bP_{0,x}(t) \leq
    \bP_{0,x}(t+1)$ where $x\in\{0,1\}^d$ with $|x|\geq d/2$.
    First note that for all $j \in [|x|,d]$,
    $\sum_{k=j}^{d} \simplePr[$exactly $k$ coordinates chosen in $t$ steps$]\leq \sum_{k=j}^{d} \simplePr[$exactly $k$ coordinates chosen in $t+1$ steps$]$
    since the distribution of chosen coordinates after $t+1$ steps clearly dominates the distribution of chosen coordinates after $t$ steps.
    Observe that for any $|x| \geq d/2$ the function
    $
     f(k) := 2^{-k} \, \binom{d-|x|}{k-|x|} \big/ \binom{d}{k}
    $
    is monotone increasing in $|x| \leq k \leq d$.
    This can be verified by showing that $f(k)/f(k-1) \geq 1$ for any $k$ with $|x| < k \leq d$.
    This allows us to apply \lemref{basiclemma2} to
    \eq{balls},
    giving $\bP_{0,x}(t) \leq \bP_{0,x}(t+1)$.
    Hence the proposition follows.
\end{proof}


\begin{thebibliography}{39}
\providecommand{\natexlab}[1]{#1}
\providecommand{\url}[1]{\texttt{#1}}
\expandafter\ifx\csname urlstyle\endcsname\relax
  \providecommand{\doi}[1]{doi: #1}\else
  \providecommand{\doi}{doi: \begingroup \urlstyle{rm}\Url}\fi

\bibitem[Aiello et~al.(1993)Aiello, Awerbuch, Maggs, and Rao]{AielloAMR93}
W.~Aiello, B.~Awerbuch, B.~M. Maggs, and S.~Rao.
\newblock Approximate load balancing on dynamic and asynchronous networks.
\newblock In \emph{\STOC{25th}{93}}, pages 632--641, 1993.

\bibitem[Barve et~al.(1997)Barve, Grove, and Vitter]{BarveGV97}
R.~D. Barve, E.~F. Grove, and J.~S. Vitter.
\newblock Simple randomized mergesort on parallel disks.
\newblock \emph{Parallel Computing}, 23\penalty0 (4-5):\penalty0 601--631,
  1997.

\bibitem[Chung(2006)]{Ch92}
F.~Chung.
\newblock \emph{Spectral Graph Theory}.
\newblock AMS, revised edition, 2006.
\newblock URL \url{http://www.math.ucsd.edu/$\sim$fan/research/revised.html}.

\bibitem[Cooper and Spencer(2006)]{CooperSpencer}
J.~Cooper and J.~Spencer.
\newblock Simulating a random walk with constant error.
\newblock \emph{Combinatorics, Probability \& Computing}, 15:\penalty0
  815--822, 2006.

\bibitem[Cybenko(1989)]{Cyb89}
G.~Cybenko.
\newblock Load balancing for distributed memory multiprocessors.
\newblock \emph{Journal of Parallel and Distributed Computing}, 7:\penalty0
  279--301, 1989.

\bibitem[Diaconis et~al.(1990)Diaconis, Graham, and Morrison]{DGM90}
P.~Diaconis, R.~Graham, and J.~Morrison.
\newblock Asymptotic analysis of a random walk on a hypercube with many
  dimensions.
\newblock \emph{Random Structures and Algorithms}, 1\penalty0 (1):\penalty0
  51--72, 1990.

\bibitem[Doerr et~al.(2008)Doerr, Friedrich, and Sauerwald]{DFS08}
B.~Doerr, T.~Friedrich, and T.~Sauerwald.
\newblock Quasirandom rumor spreading.
\newblock In \emph{\SODA{19th}{08}}, pages 773--781, 2008.

\bibitem[Els\"{a}sser and Sauerwald(2010)]{ES10}
R.~Els\"{a}sser and T.~Sauerwald.
\newblock Discrete load balancing is (almost) as easy as continuous load
  balancing.
\newblock In \emph{\PODC{29th}{10}}, pages 346--354, 2010.

\bibitem[Els\"asser et~al.(2006)Els\"asser, Monien, and Schamberger]{EMS06}
R.~Els\"asser, B.~Monien, and S.~Schamberger.
\newblock Distributing unit size workload packages in heterogenous networks.
\newblock \emph{Journal of Graph Algorithms \& Applications}, 10\penalty0
  (1):\penalty0 51--68, 2006.

\bibitem[Fill(2009)]{Fill2009a}
J.~A. Fill.
\newblock The passage time distribution for a birth-and-death chain: Strong
  stationary duality gives a first stochastic proof.
\newblock \emph{Journal of Theoretical Probability}, 22\penalty0 (3):\penalty0
  543--557, 2009.

\bibitem[Friedrich and Sauerwald(2009)]{FS09}
T.~Friedrich and T.~Sauerwald.
\newblock Near-perfect load balancing by randomized rounding.
\newblock In \emph{\STOC{41st}{09}}, pages 121--130, 2009.

\bibitem[Friedrich et~al.(2010)Friedrich, Gairing, and Sauerwald]{FGS10}
T.~Friedrich, M.~Gairing, and T.~Sauerwald.
\newblock Quasirandom load balancing.
\newblock In \emph{\SODA{21st}{10}}, pages 1620--1629. SIAM, 2010.

\bibitem[Friedrich et~al.(2012)Friedrich, Gairing, and Sauerwald]{SICOMP1}
T.~Friedrich, M.~Gairing, and T.~Sauerwald.
\newblock Quasirandom load balancing.
\newblock \emph{SIAM Journal on Computing}, 41\penalty0 (4):\penalty0 747--771,
  2012.

\bibitem[Gehrke et~al.(1999)Gehrke, Plaxton, and Rajaraman]{GPR99}
J.~Gehrke, C.~Plaxton, and R.~Rajaraman.
\newblock Rapid convergence of a local load balancing algorithm for
  asynchronous rings.
\newblock \emph{Theoretical Computer Science}, 220\penalty0 (1):\penalty0
  247--265, 1999.

\bibitem[Ghosh and Muthukrishnan(1996)]{GM96}
B.~Ghosh and S.~Muthukrishnan.
\newblock Dynamic load balancing by random matchings.
\newblock \emph{Journal of Computer and System Sciences}, 53\penalty0
  (3):\penalty0 357--370, 1996.

\bibitem[Ghosh et~al.(1999)Ghosh, Leighton, Maggs, Muthukrishnan, Plaxton,
  Rajaraman, Richa, Tarjan, and Zuckerman]{GhoshLMMPRRTZ99}
B.~Ghosh, F.~T. Leighton, B.~M. Maggs, S.~Muthukrishnan, C.~G. Plaxton,
  R.~Rajaraman, A.~W. Richa, R.~E. Tarjan, and D.~Zuckerman.
\newblock Tight analyses of two local load balancing algorithms.
\newblock \emph{SIAM Journal on Computing}, 29\penalty0 (1):\penalty0 29--64,
  1999.

\bibitem[Grimmet and Stirzaker(2001)]{GS01}
G.~Grimmet and D.~Stirzaker.
\newblock \emph{Probability and Random Processes}.
\newblock Oxford University Press, 3rd edition, 2001.

\bibitem[Guruswami(2000)]{Gur00}
V.~Guruswami.
\newblock Rapidly {M}ixing {M}arkov {C}hains, 2000.
\newblock Unpublished, available at
  \url{www.cs.cmu.edu/$\sim$venkatg/pubs/papers/markov-survey.ps}.

\bibitem[Hoggar(1974)]{Hoggar1974}
S.~G. Hoggar.
\newblock Chromatic polynomials and logarithmic concavity.
\newblock \emph{Journal of Combinatorial Theory, Series B}, 16\penalty0
  (3):\penalty0 248 -- 254, 1974.

\bibitem[Jan and Hwang(2003)]{JH03}
G.~Jan and Y.~Hwang.
\newblock An efficient algorithm for perfect load balancing on hypercube
  multiprocessors.
\newblock \emph{The Journal of Supercomputing}, 25\penalty0 (1):\penalty0
  5--15, 2003.

\bibitem[Karlin and McGregor(1959)]{KarlinMcGregor59}
S.~Karlin and J.~McGregor.
\newblock Coincidence properties of birth and death processes.
\newblock \emph{Pacific J. Math.}, 9\penalty0 (4):\penalty0 1109--1140, 1959.

\bibitem[Karlin et~al.(1993)Karlin, Lindqvist, and Yao]{KLY93}
S.~Karlin, B.~Lindqvist, and Y.-C. Yao.
\newblock Markov chains on hypercubes: Spectral representations and several
  majorization relations.
\newblock \emph{Random Structures \& Algorithms}, 4\penalty0 (1):\penalty0
  1--36, 1993.

\bibitem[Keilson(1979)]{Keilson}
J.~Keilson.
\newblock \emph{Markov Chain Models -- Rarity and Exponentiality}.
\newblock Springer Verlag, 1979.

\bibitem[Keilson and Gerber(1971)]{KeilsonGerber1971}
J.~Keilson and H.~Gerber.
\newblock Some results for discrete unimodality.
\newblock \emph{Journal of the American Statistical Association}, 66\penalty0
  (334):\penalty0 386--389, 1971.

\bibitem[Lawler(1991)]{Lawler}
G.~F. Lawler.
\newblock \emph{Intersections of random walks}.
\newblock Probability and its Applications. Birkh\"auser, 1991.

\bibitem[Lawler and Limic(2010)]{LawlerLimic}
G.~F. Lawler and V.~Limic.
\newblock \emph{Random Walk: A Modern Introduction}.
\newblock Cambridge Studies in Advanced Mathematics. Cambridge University
  Press, 2010.

\bibitem[Lov\'{a}sz(1993)]{Lovasz93random}
L.~Lov\'{a}sz.
\newblock Random walks on graphs: A survey.
\newblock \emph{Combinatorics, Paul Erd{\H{o}}s is Eighty}, 2:\penalty0 1--46,
  1993.

\bibitem[Lov\'{a}sz and Winkler(1995)]{LovaszWinkler95}
L.~Lov\'{a}sz and P.~Winkler.
\newblock Mixing of random walks and other diffusions on a graph.
\newblock \emph{Surveys in combinatorics}, pages 119--154, 1995.

\bibitem[Morris and Sinclair(2004)]{MS04}
B.~Morris and A.~Sinclair.
\newblock Random walks on truncated cubes and sampling 0-1 knapsack solutions.
\newblock \emph{SIAM Journal on Computing}, 34\penalty0 (1):\penalty0 195--226,
  2004.

\bibitem[Motwani and Raghavan(1995)]{MR95}
R.~Motwani and P.~Raghavan.
\newblock \emph{Randomized Algorithms}.
\newblock Cambridge University Press, 7th edition, 1995.

\bibitem[Muthukrishnan et~al.(1998)Muthukrishnan, Ghosh, and
  Schultz]{MuthukrishnanGS98}
S.~Muthukrishnan, B.~Ghosh, and M.~H. Schultz.
\newblock First- and second-order diffusive methods for rapid, coarse,
  distributed load balancing.
\newblock \emph{Theory of Computing Systems}, 31\penalty0 (4):\penalty0
  331--354, 1998.

\bibitem[Plaxton(1989)]{P89}
C.~Plaxton.
\newblock Load balancing, selection sorting on the hypercube.
\newblock In \emph{\SPAA{1st}{89}}, pages 64--73, 1989.

\bibitem[Rabani et~al.(1998)Rabani, Sinclair, and Wanka]{RSW98}
Y.~Rabani, A.~Sinclair, and R.~Wanka.
\newblock Local divergence of {M}arkov chains and the analysis of iterative
  load balancing schemes.
\newblock In \emph{\FOCS{39th}{98}}, pages 694--705, 1998.

\bibitem[Ross(2007)]{Ro07}
S.~Ross.
\newblock \emph{Introduction to Probability Models}.
\newblock Academic Press, 2007.

\bibitem[{Stanley}(1989)]{Stanley1989}
R.~P. {Stanley}.
\newblock Log-concave and unimodal sequences in algebra, combinatorics, and
  geometry.
\newblock \emph{New York Academy Sciences Annals}, 576:\penalty0 500--535, Dec.
  1989.

\bibitem[Subramanian and Scherson(1994)]{SubramanianScherson94}
R.~Subramanian and I.~D. Scherson.
\newblock An analysis of diffusive load-balancing.
\newblock In \emph{\SPAA{6th}{94}}, pages 220--225, New York, NY, USA, 1994.
  ACM.

\bibitem[Surana et~al.(2006)Surana, Godfrey, Lakshminarayanan, Karp, and
  Stoica]{Surana06}
S.~Surana, B.~Godfrey, K.~Lakshminarayanan, R.~Karp, and I.~Stoica.
\newblock Load balancing in dynamic structured peer-to-peer systems.
\newblock \emph{Performance Evaluation}, 63\penalty0 (3):\penalty0 217--240,
  2006.

\bibitem[Williams(1991)]{Williams91}
R.~D. Williams.
\newblock Performance of dynamic load balancing algorithms for unstructured
  mesh calculations.
\newblock \emph{Concurrency: Practice and Experience}, 3\penalty0 (5):\penalty0
  457--481, 1991.

\bibitem[Zhanga et~al.(2009)Zhanga, Jianga, and Li]{Zhanga09}
D.~Zhanga, C.~Jianga, and S.~Li.
\newblock A fast adaptive load balancing method for parallel particle-based
  simulations.
\newblock \emph{Simulation Modelling Practice and Theory}, 17\penalty0
  (6):\penalty0 1032--1042, 2009.

\end{thebibliography}
\end{document}